\documentclass[journal]{IEEEtran}
\usepackage{array}

\usepackage{mdwmath}
\usepackage{mdwtab}

\usepackage[dvips]{graphicx}

\usepackage[T1]{fontenc}
\usepackage{lmodern}

\usepackage[tight,footnotesize]{subfigure}
\usepackage{caption}
\usepackage{epstopdf}

\usepackage{float}
\usepackage{stfloats}
\usepackage{lipsum}

\usepackage{tabularx}

\usepackage{amsmath}
\usepackage{amssymb}
\usepackage{paralist}
\usepackage{mathrsfs}
\usepackage{url}
\usepackage{mathbbol}
\usepackage{times}
\usepackage{booktabs}
\usepackage{verbatim}

\usepackage{bm}			

\usepackage{amsthm}
\newtheorem{theorem}{Theorem}
\newtheorem{definition}{Definition}

\newtheorem{lemma}{Lemma}

\newtheorem{corollary}{Corollary}
\newtheorem{assumption}{Assumption}
\newtheorem{result}{Result}

\usepackage{color}

\newcommand{\REMOVED}[1]{}

\newcommand{\eq}[1]{Eq.~\eqref{#1}}

\newcommand{\Np}{N_{p}^{(\mathcal{M})}}
\newcommand{\Na}{N_{a}^{(\mathcal{M})}}
\newcommand{\Cp}{\mathcal{C}_{p}^{(\mathcal{M})}}
\newcommand{\Ca}{\mathcal{C}_{a}^{(\mathcal{M})}}
\newcommand{\XM}{{X}_{\mathcal{M}}}

\newcommand{\revisionRed}[1]{{#1}}
\newcommand{\red}[1]{{#1}}

\usepackage[normalem]{ulem}
\newcommand{\redout}{\bgroup\markoverwith{\textcolor{red}{\rule[.3ex]{2pt}{1pt}}}\ULon}

\usepackage{algorithm}
\usepackage{algorithmicx}
\usepackage{algpseudocode}

\begin{document}
\title{Effects of Content Popularity on the Performance of Content-Centric Opportunistic Networking:\\An Analytical Approach and Applications}
%
%
%

\author{Pavlos~Sermpezis, 
        Thrasyvoulos Spyropoulos,~\IEEEmembership{Member,~IEEE}
\thanks{P. Sermpezis and T. Spyropoulos are with the Department
of Mobile Communications, EURECOM, France, e-mail: firstname.lastname@eurecom.fr.}
\thanks{}
\thanks{A preliminary version of this paper appeared in Proc. 15th ACM International Symposium on Mobile Ad Hoc Networking and Computing (MobiHoc), 2014~\cite{pavlos-not-all-content}.}}

%

\maketitle

\begin{abstract}
Mobile users are envisioned to exploit direct communication opportunities between their portable devices, in order to enrich the set of services they can access through cellular or WiFi networks. Sharing contents of common interest or providing access to resources or services between peers can enhance a mobile node's capabilities, offload the cellular network, and disseminate information to nodes without Internet access. Interest patterns, i.e. how many nodes are interested in each content or service (popularity), as well as how many users can provide a content or service (availability) impact the performance and feasibility of envisioned applications. In this paper, we establish an analytical framework to study the effects of these factors on the delay and success probability of a content/service access request through opportunistic communication. We also apply our framework to the mobile data offloading problem and provide insights for the optimization of its performance. We validate our model and results through realistic simulations, using datasets of real opportunistic networks.
\end{abstract}



%
\IEEEpeerreviewmaketitle

\section{Introduction}

\IEEEPARstart{O}{pportunistic} or Delay Tolerant Networks (DTNs) consist of mobile devices (e.g. smartphones, laptops) that can exchange data using direct communication (e.g. Bluetooth, WiFi Direct) when they are within transmission range.  While initially proposed for communication in extreme environments, the proliferation of ``smart'' mobile devices has led researchers to consider opportunistic networks as a way to support existing infrastructure and/or novel applications, like file sharing~\cite{Gao-user-centric-DTN, Yoneki-publish-subscribe-dtn}, crowd sensing~\cite{MobComp-next-decade,Ott-oppnet-applications}, collaborative computing~\cite{opp-computing,Scampi-OppComp}, offloading of cellular networks~\cite{offloading-wowmom11,Hui-Offloading, multiple-offloading}, etc.

This trend is also shifting the focus from end-to-end to \textit{content-centric} communications. \red{In a content-centric application some nodes $\geq1$ (the \textit{holders}) have the \textit{same} content item (e.g. a data file, a service), and some nodes $\geq1$ (the \textit{requesters}) are interested in this content. The goal of the communication mechanism is the requesters to get the content from the holders}. Some content-centric applications for which opportunistic networking has been considered are: (i) \textit{content sharing}~\cite{Gao-user-centric-DTN,contentplace,CEDO}: the source(s) of a content (e.g. multimedia file, web page) might want to distribute it (e.g. user generated content) or is willing to share it with other nodes (e.g. content downloaded earlier); (ii) \textit{service or resource access}~\cite{opp-computing,Scampi-OppComp}: nodes offer access to resources (e.g. Internet access) or services (e.g. computing resources); (iii) \textit{mobile data offloading}~\cite{offloading-wowmom11,Hui-Offloading,multiple-offloading}: the cellular network provider, instead of serving separately each node requesting a given content (e.g. a popular video, or software update), distributes a few copies of the content in some relay nodes (\textit{holders}) and they can further forward it to any other node that makes a request for it. 

The performance of these mechanisms highly depends on \textit{who} is interested, in \textit{what}, and \textit{where} it can be found (i.e. which other nodes have it). While the effect of node mobility has been extensively considered (e.g.~\cite{Gao-user-centric-DTN, contentplace, Picu:wowmom2012}) content popularity has been mainly considered from an algorithmic perspective (e.g~\cite{multiple-offloading,CEDO}), and in the context of a specific application. Despite the inherent interest of these studies, some questions remain:  Would a given allocation policy work well in a different network setting? Are there interest patterns that would make a scheme generally better than others? Key factors like content popularity and content availability might impact the performance or even decide the feasibility of a given application altogether. In this paper, we try to provide some initial insight into these questions, by contributing along the following key directions: 

\revisionRed{
\textbf{Content popularity model.} We propose an analytical framework that is applicable to a range of mobility and content popularity patterns seen in real networks (Section~\ref{sec:network-model}). Its \textit{simplicity} and \textit{generality} can render it a useful tool for future modeling/analytic studies. To our best knowledge, this is the first application-independent effort in this direction.
}

\revisionRed{\textbf{Performance analysis.} We derive closed form expressions for the prediction of important performance metrics (Section~\ref{sec:analysis}). We first derive exact predictions and bounds for the performance of content delivery in a base scenario, and then extend our analysis to more generic mobility and traffic cases.}

\revisionRed{The practicality of our results lies in the fact that only a few statistics about the aggregate mobility and content popularity patterns is needed. Hence, they facilitate \textit{online} performance prediction and protocol tuning, compared to approaches (as e.g.~\cite{multiple-offloading}) requiring detailed per node statistics that are hard to acquire in real scenarios. Moreover, they can complement system design or feasibility studies. The presented simulations, which validate the accuracy of the theoretical predictions, indicate how a sensitivity analysis of system parameters and a comparison of different mechanisms can be performed.}

\revisionRed{
\textbf{Mobile data offloading optimization.} While a detailed application-specific optimization is beyond the scope of this paper, we demonstrate how our framework can be applied to an example application: mobile data offloading  (Section~\ref{sec:applications}). Using our analysis, we show how an offloading mechanism can be optimized and discuss what the performance related implications are. This case study provides guidelines to researchers for investigating and similarly proceeding in the analysis and optimization of further content-centric applications, policies, protocols, etc.
}

\revisionRed{
Finally, we discuss related work in Section~\ref{sec:related-work}, and conclude our paper in Section~\ref{sec:conclusion}.
}

\section{Network Model}\label{sec:network-model}
\subsection{Mobility Model}\label{sec:mobility-model}
We consider a network $\mathcal{N}$, where $N$ nodes move in an area, much larger than their transmission range. Data packet exchanges between a pair of nodes can take place only when they are in proximity (\textit{in contact}). Hence, the time points, when the contact events take place, and the nodes involved, determine the dissemination of a message.

We assume that the sequence of the contact events between nodes $i$ and $j$ is given by a random point process with rate $\lambda_{ij}$\footnote{We ignore the contact duration and assume infinite bandwidth; assumptions that are common (e.g.~\cite{Gao-user-centric-DTN,multiple-offloading}) and orthogonal to the problem we consider here.}. Analyses of real-world traces suggest that the times between consecutive contacts for a given pair can often be approximated (completely or in the tail) as either exponentially~\cite{Gao2009,Conan2007} or power-law (e.g. Pareto) distributed~\cite{chaintreau-tmc}. \red{Our analysis can be applied to both cases, as well as for other distribution types. In the remainder, we focus on the exponential inter-contact times case, which can be described with a single main parameter $\lambda_{ij}$ (the \textit{contact rate}), and we further demonstrate its applicability to a simple Pareto inter-contact times case.}

The network $\mathcal{N}$ can be described with the contact (or meeting) rates matrix $\mathbf{\Lambda} = \lbrace\lambda_{ij}\rbrace$. Depending on the underlying mobility process, there might be large differences between the different $\lambda_{ij}$ values in this matrix. Furthermore, it is often quite difficult, in a DTN context, to know $\mathbf{\Lambda}$ exactly, or estimates might be rather noisy. For these reasons, we consider the following simple model for $\mathbf{\Lambda}$:

\begin{assumption}\label{ass:heterogeneous-mobility}
 The contact rates $\lambda_{ij}$ are drawn from an arbitrary distribution with probability density function $f_{\lambda}(\lambda)$ with known mean $\mu_{\lambda}$ and variance $\sigma^{2}_{\lambda}$ ($CV_{\lambda} = \frac{\sigma_{\lambda}}{\mu_{\lambda}}$).
\end{assumption}

By choosing the right function $f_{\lambda}$ the above model can capture heterogeneity in the pairwise contact rates, or noise in the estimates. In practice, one would fit the empirical distribution observed in a given measurement trace with an $\hat{f}_{\lambda}$ and use it in the analysis\footnote{\red{In some scenarios node mobility might have some further, more complex characteristics, e.g. node pairs contact with different frequency $\lambda_{ij}$ during day/night or weekdays/weekends. However, in the majority of applications, it can be safely assumed that there is a \textit{time-scale separation}, i.e., the time to deliver a content is much smaller than a period of similar $\lambda_{ij}$ values. Moreover, a real system can use different rates depending on the considered period, e.g. $\lambda_{ij}^{(day)}$ and $\lambda_{ij}^{(night)}$, or a running estimate mechanism.}}.

\revisionRed{Summarizing, the above model is a trade-off between \textit{realism}, \textit{analytical tractability}, and \textit{usefulness}. Our choice, among the several options of the aforementioned trade-off, is motivated as follows. The above model for  $\mathbf{\Lambda} = \lbrace\lambda_{ij}\rbrace$ can capture many aspects of the contact rates' heterogeneity (i.e., different pair-wise rates $\lambda_{ij}$ and distributions $f_{\lambda}(\lambda)$). At the same time, it is a probabilistic model and, thus, it remains simple enough to derive insightful, closed-form results for the performance of a content delivery (Section~\ref{sec:analysis}), which is the main goal of this paper. Finally, if more detailed mobility characteristics (e.g., temporal or periodic patterns~\cite{transient-contact-patterns}) were needed to predict performance, this would make our results less useful for designing a system/application, since in a real scenario it is not always possible (or practical) to acquire all this information (or, at least, in real time).}

\subsection{Content Traffic Model}\label{sec:content-traffic-model}

We assume that each node might be \textit{interested in} one or more ``contents''. A content of interest might refer to  (i) a single piece of data (e.g. a multimedia file, a google map)~\cite{offloading-wowmom11}, (ii) all messages/data belonging to a category of interests (e.g. local events, financial news)~\cite{Yoneki-publish-subscribe-dtn,Costa-publish-subscribe-dtn}, (iii) updates and feeds (e.g. weather forecast, latest news)~\cite{podcasting}, etc. 

A number of content-sharing applications and mechanisms have been proposed in previous literature, from publish-subscribe mechanisms to ``channel''-based sharing and device-to-device offloading, etc., (e.g.~\cite{Yoneki-publish-subscribe-dtn,MobComp-next-decade,Ott-oppnet-applications,podcasting}). To proceed with our analysis we need to setup a model of content/service access. \red{In the following, we propose a generic model for content-centric applications.}

The main notation we use in our model and analysis is summarized in Table~\ref{table:notation}.

\subsection*{\textbf{Content Popularity}}

We assume that when a node is interested in a content or service, it queries other nodes it \textit{directly} encounters for it.  We denote the event that a node $i\in\mathcal{N}$ is \textit{interested in} a content $\mathcal{M}$ (or, equivalently, $i$ requests $\mathcal{M}$) as: $i\rightarrow \mathcal{M}.$ We further denote the set of all the contents that nodes are interested in, as: $\textbf{M} = \lbrace \mathcal{M}: \exists i\in\mathcal{N}, i\rightarrow\mathcal{M} \rbrace$. $| \textbf{M} | = M$, where $|\cdot |$ denotes the cardinality of a set.
\begin{definition}[Content Popularity]\label{def:content-popularity}
We define the popularity of a content $\mathcal{M}$ as the number of nodes $\Np$ that are interested in it\footnote{This could be an average, calculated over some time window.}: 
\begin{equation}\label{eq:popularity-definition}
\Np =  |\mathcal{C}_{p}^{(\mathcal{M})}|, \text{ where } \mathcal{C}_{p}^{(\mathcal{M})} = \lbrace i\in \mathcal{N}: i\rightarrow\mathcal{M}\rbrace 
\end{equation}
We further denote the percentage of contents with a given popularity value $n$ as 
\begin{equation}\label{eq:popularity-distribution}
 P_{p}(n) = \frac{1}{M}\sum_{\mathcal{M}\in\textbf{M}}\mathcal{I}_{\Np=n},~~~ n\in[0,N]
\end{equation}
where $\mathcal{I}_{\Np=n}=1$ when $\Np=n$ and $0$ otherwise. 
\end{definition}

In other words, $P_{p}(n)$ defines a probability distribution over the different contents and associated popularities: \red{If we randomly choose one content $\mathcal{M}\in\textbf{M}$, then the probability that its popularity is equal to $n$ is given by $P_{p}(n)$.}

\red{In practice, the content popularity distribution $P_{p}(n)$ might be known exactly, estimated, or predicted, depending on the given scenario and application. For instance, in a publish-subscribe application, users subscribe in advance in different channels, and thus the popularity of each channel/content can be known or estimated through distributed mechanisms. In a mobile data offloading scenario, the cellular network might be informed from users about their requests, or infer popularity from their interest profiles~\cite{multiple-offloading}. In content sharing application, the popularity of a file can be predicted using methods based, e.g., on past statistics,
early demand of a content, social dynamics, etc.~\cite{thesis-popularity-prediction}.}

\subsection*{\textbf{Content Availability}}
We assume that a request for a content or service is completed, when (and if) a node that holds (a copy of) the requested content is \textit{directly} encountered. We denote the event that a node $i$ holds (a copy of) a content $\mathcal{M}$ as $i \leftarrow\mathcal{M}$, 
and we define the availability $\Na$ of a content $\mathcal{M}$ as
\begin{definition}[Content Availability]\label{def:content-avalability}
The availability of a content message $\mathcal{M}$ is defined as the number of nodes $N_{a}^{(\mathcal{M})}$ that hold a copy of it. 
\begin{equation}
\Na =  |\mathcal{C}_{a}^{(\mathcal{M})}|, \text{ where } \mathcal{C}_{a}^{(\mathcal{M})} = \lbrace i\in \mathcal{N}: i\leftarrow\mathcal{M}\rbrace
\end{equation}
\end{definition}

The availability of a given content might often (although not always) be correlated with the popularity of that content. A cellular network provider, for example, might \textit{allocate} more holders for popular contents~\cite{multiple-offloading}. In a content-sharing setting, where some nodes might be more willing than others to maintain and share (``seed'') a content after they have downloaded and ``consumed'' it, popular content will end up being shared by more nodes. We will model such correlations in a \textit{probabilistic} way, as follows.
\begin{definition}[Availability vs. Popularity]\label{def:avail-popul-relation-generic}
The availability of a content item is related to its popularity through the relation
\begin{equation}\label{eq:g_mn}
\red{P\lbrace N_{a} = m | N_{p} = n \rbrace = g(m|n)}
\end{equation}
\end{definition}
The above conditional probabilities can describe a wide range of cases where availability depends on popularity, and some additional randomness might be present due to factors like: natural churn in the nodes sharing the content, content-dependent differences in the sharing policies applied by nodes, estimation noise, etc. \red{For example, we might assume that a content of higher popularity has \emph{on average} higher availability, but the actual availability (e.g. over a given time window) is subject to some randomness due to node churn, etc.} 

Some special cases of this model include: 

(i) \textit{Uncorrelated availability}, where $g(m|n) \equiv g(m)$. \red{For example, in service/resource access applications, where holders are the nodes that can provide access to some resources (e.g. Internet access, software)~\cite{Scampi-OppComp}, the availability depends on the number of devices with the given resources rather than the number of users that are interested in them.}

(ii) \textit{Deterministic availability}, where:
\begin{equation} N_{a} = \rho\left(N_{p}\right) ~~\Leftrightarrow~~ 
 g(m|n) = \bigg\{
\begin{tabular}{l}
$1,~~m=\rho(n)$\\
$0,~~\text{otherwise}$
\end{tabular}\nonumber
\end{equation}
where $\rho(n): [1,N] \rightarrow [0,N]$ can be an arbitrary function. \red{This case corresponds to applications where the number of holders is selected (by a centralized authority, a distributed protocol, etc.) according to the popularity of a content. Note also that this deterministic formula $\rho(n)$ can be used as an approximation of the general case, where the noise around the mean is ignored (i.e., only the mean value, rather than the exact distribution, of the availability needs to be known/estimated):
\[\rho(n) = \bar{g}(n) \equiv \sum_{m} m\cdot g(m|n)\]}

\begin{table}
 \centering
 \caption{Important Notation}
\begin{scriptsize}
\begin{tabular}{|l|l|c|}
\hline 
\multicolumn{3}{|l|}{MOBILITY (Section~\ref{sec:mobility-model})}\\
\hline
$\lambda_{ij}$	& Contact rate between nodes $i$ and $j$ &{}\\
\hline
$f_{\lambda}(\lambda)$		& Contact rates distribution & {}\\
\hline
$\mu_{\lambda}, \sigma_{\lambda}^{2}$	& Mean value/ variance of contact rates, $CV_{\lambda} = \frac{\sigma_{\lambda}}{\mu_{\lambda}}$& {}\\
\hline
\multicolumn{3}{|l|}{CONTENT TRAFFIC (Section~\ref{sec:content-traffic-model})}\\
\hline
$i\rightarrow\mathcal{M}$		& Node $i$ \textit{is interested} / \textit{requests} content $\mathcal{M}$ &{}\\
\hline
$\textbf{M}$		& Set of contents in the network, $| \textbf{M} | = M$. &{}\\
\hline
$\Np$	& Popularity of content $\mathcal{M}$ &{Def.~\ref{def:content-popularity}}\\
\hline
$\Cp$		& Set of nodes interested in content $\mathcal{M}$ &{Def.~\ref{def:content-popularity}}\\
\hline
$P_{p}(n)$		& Probability distribution of content popularity &\eq{eq:popularity-distribution} \\
\hline
$i \leftarrow\mathcal{M}$ &  Node $i$ \textit{holds a copy} of content $\mathcal{M}$ &{}\\
\hline
$\Na$	& Availability of content $\mathcal{M}$ &{Def.~\ref{def:content-avalability}}\\
\hline
$\Ca$		& Set of nodes that hold a copy of content $\mathcal{M}$ &{Def.~\ref{def:content-avalability}}\\
\hline
$g(m|n)$	& Availability - Popularity relation &{Def.~\ref{def:avail-popul-relation-generic}}\\
\hline
$\rho(n)$	& Deterministic case for $g(m|n)$&{}\\
\hline
$\overline{g}(n)$	&The average value of $g(\cdot|n)$&{}\\
\hline
\multicolumn{3}{|l|}{ANALYSIS (Section~\ref{sec:preliminary-analysis})}\\
\hline
$P_{p}^{req.}(n)$	& Popularity distribution of a random request&{Lemma~\ref{thm:Pint-popularity}}\\
\hline
$P_{a}^{req.}(n)$	& Availability distribution of a random request&{Lemma~\ref{thm:Pint-availability-generic}}\\
\hline
$T_{ij}$	& Time of \textit{next} meeting between nodes $i$ and $j$&{}\\
\hline
$T_{\mathcal{M}}$	& Content access time &{}\\
\hline
$\XM$	& Sum of meeting rates of $j$ and nodes $\in\Ca$&{\eq{eq:XM-definition}}\\
\hline
\end{tabular}
\end{scriptsize}
\label{table:notation}
\end{table}

\section{Analysis of Content Requests}\label{sec:analysis}
We will now analyze how different popularity, availability, and mobility patterns (possibly arising from different applications, policies, and network settings) affect performance metrics like: (i) the delay to access a content of interest, (ii) the probability to retrieve a content before a  deadline. A key parameter for these metrics is the number of holders for the requested content (availability). The higher this number, the sooner a requesting node will encounter one of them.  

While content availability might sometimes be time dependent~\cite{CEDO}, or the content holders might be chosen based on their mobility properties~\cite{multiple-offloading}, \red{as a first step we make two additional, restrictive assumptions that allow us to derive simple, useful expressions. Later, in Section~\ref{sec:analysis-extensions}, we relax both these assumptions, and show how our analysis and results can be modified to capture more generic scenarios where availability can be dependent on the time (or the content dissemination process) and mobility patterns.}

\begin{assumption}\label{ass:availability-time-independence}
The (i) popularity $\Np$ and (ii) availability $\Na$ of a content $\mathcal{M}$ do not change over time.
\end{assumption}

\begin{assumption}\label{ass:traffic-mobility-independence}
The sets of requesters $\Cp$ and holders $\Ca$ of a content $\mathcal{M}$ are independent of node mobility.
\end{assumption} 

\red{Regarding the validity of Assumption~\ref{ass:availability-time-independence}, it can be safely assumed that users' interests do not change, at least in the time window of a content delivery. This is a common assumption in the majority of related works. As a result, \textit{content popularity}, which is given by the number of the nodes interested in a content, is not expected to change as well.}

\red{With respect to \textit{content availability}, the assumption is valid (or a good approximation) in a number of applications. For example, in the case that the number of holders is chosen by the cellular operator~\cite{Hui-Offloading,multiple-offloading} or content provider, and other nodes cannot act as holders or do not have incentives to do so. It is also valid when the ``content" is a service (e.g. Internet access, or specific sensor) that is offered only by a certain number of devices~\cite{Scampi-OppComp}. Moreover, in content sharing applications / protocols where users have a limited ``budget" of $L$ copies that can distribute to relay nodes (i.e. the holders), if $L\equiv \Na \ll N$ and $\Np\ll N$ (which is reasonable for a typical opportunistic networking scenario), then the time of the initial content distribution to holders is much less than the time needed by a requester to access the content\footnote{\revisionRed{In a simple example, of a network with $N=1000$ nodes and $\lambda_{ij}=\lambda$, a ``source'' node originally has a content, in which $\Np=20$ nodes are interested. The ``source'' replicates the content to the first $L=4$ nodes it meets (cf. \textit{source Spray and Wait} protocol~\cite{spray-and-wait}), which act as holders. Then, it can be easily shown that the expected time till \textit{all} holders get the content is $E[T_{s}] \approx \frac{\Na}{N\cdot \lambda} = \frac{1}{250\cdot \lambda}$, while the expected time a (i.e., \textit{any}) requester to access the content is $E[T_{a}] \approx \frac{1}{\Na\cdot\lambda} = \frac{1}{(L+1)\cdot\lambda}= \frac{1}{5\lambda}$, or equivalently $E[T_{a}] = 50\cdot E[T_{s}]\gg E[T_{s}]$.}}. Hence, considering only the time of the content sharing process \textit{after} the initial distribution to holders, the condition for time-invariant availability holds.}

\red{Nevertheless, in scenarios where a content is disseminating and new nodes (e.g. the requesters after receiving it) are willing to share it~\cite{offloading-wowmom11}, then the availability might change over time. We consider and analyze such cases in Section~\ref{sec:analysis-extensions}, as an extension of our basic results of Section~\ref{sec:performance-metrics}.}

\red{Assumption~\ref{ass:traffic-mobility-independence} holds when a \textit{mobility oblivious} allocation policy (i.e. randomized protocols) is considered, e.g.~\cite{CEDO}, or the homogeneous algorithm of~\cite{multiple-offloading}. It is also a reasonable approximation, in settings where there is no knowledge of the interests-mobility correlation, if any.}

\red{Nevertheless, there exist scenarios where who holds what content might depend on the contact rates with other nodes (i.e. the mobility), and such a dependence can possibly affect the performance. This dependence might occur due to the employed dissemination protocol~\cite{contentplace, multiple-offloading}. In fact, many protocols proposed in related literature, try to exploit mobility or social characteristics of nodes, in order to find a set of holders that contact regularly the requesters and can, thus, deliver the contents to them in a fast and efficient way.} 

\red{However, due to the different mechanisms employed, a different (and very complex in some cases) analytic approach would be needed for each protocol. To this end, in Section~\ref{sec:analysis-extensions}, we \textit{do} take into account mobility-aware schemes, in a generic and application-independent way. Furthermore, with this proposed extension of our model, one can capture scenarios where mobility-availability correlation do not come (necessarily) from a dissemination protocol, but they exist due to some underlying heterogeneous traffic patterns~\cite{pavlos-TMC-traffic}.}

\subsection{Preliminary Analysis}\label{sec:preliminary-analysis}

\red{Assume a content-centric application with many different contents. To predict the performance of such a system, we would like to know how long the \textit{average} request takes to be satisfied. To do so, let us pick some random user request (over all the requests made for different contents), and let us assume that this request is for some content $\mathcal{M}$\footnote{\red{We stress here that we do not refer to a certain content, but we denote the \textit{content related to the random request} as $\mathcal{M}$ for ease of reference. Hence, in the remainder we do not use the superscript $\mathcal{M}$ in the related notation. E.g. we denote the popularity as $N_{p}$ instead of $\Np$.}}.}

We first need to answer the following two questions:

\begin{enumerate}
\item[\textbf{Q.1}] What is the popularity of $\mathcal{M}$?
\item[\textbf{Q.2}] How fast does a requesting node meet $\mathcal{M}$'s holders?
\end{enumerate}

Q.1 is needed to predict the availability for the content of the random request. Given this availability, Q.2 will estimate the (sum of) contact rates between the requesting node and the holders, according to Assumptions~\ref{ass:heterogeneous-mobility} and~\ref{ass:traffic-mobility-independence}. \red{The contact rates between the requester and the holders will be then used (Section~\ref{sec:performance-metrics}) in calculating how fast the request will be satisfied.}

~\\
\noindent\underline{Answering \textbf{Q.1}}

It is easy to see that the popularity of $\mathcal{M}$ should be proportional to $P_{p}(n)$: the higher the number of different contents with a popularity value $n$, the higher the chance that $\mathcal{M}$ will be of popularity $n$. However, the higher the popularity of a content, the more the requests made for it. Hence, a first important observation is that the popularity of the content of such a \textit{random request} is \textit{not} distributed as $P_{p}(n)$ but is also proportional to the popularity value $n$.

Consider a stylized example, where only two contents exist in the network, content A with popularity value $10$ and content B with popularity value $1$. Hence, ``half'' the contents are of high popularity ($N_{p}=10$), and ``half'' of low ($N_{p}=1$), or in other words $P_{p}(10) = P_{p}(1) = \frac{1}{2}$. \red{However, there will be $10$ times more requests for content A than for content B.} Consequently, if we select a request randomly, there is a $10\times$ higher chance that it will be for content A, that is, for the content of popularity $10$. Normalizing to have a proper probability distribution gives us the following lemma.

\begin{lemma}\label{thm:Pint-popularity}
The probability that a random request is for a content of popularity equal to $n$ is given by
\begin{align}
 P_{p}^{req.}(n) = \frac{n}{E_{p}[n]}\cdot P_{p}(n)\nonumber
\end{align}
where $E_{p}[n] = \sum_{n}n\cdot P_{p}(n)$ is the average content popularity~\footnote{We use subscript $p$ to denote an expectation over the popularity distribution $P_{p}(n)$, and $n$ denotes the random popularity values.}.
\end{lemma}
\red{\underline{Remark:} We would like to mention here that in some related works, the popularity distribution is defined over the different popularities values in a set of contents, which in our framework corresponds to the distribution $P_{p}^{req.}(n)$. In contrast, we define (in a more generic way) the popularity distribution $P_{p}(n)$ as the percentage of contents having (exactly) $n$ requesters. The correspondence between the two approaches can be made using Lemma~\ref{thm:Pint-popularity}.}

\red{For the convenience of the reader, we state the following corollary that makes the aforementioned correspondence for an important example case, the Zipf-law (or discrete Pareto) distribution, which is frequently observed in real systems~\cite{youtube-traffic-from-edge, RSS-traffic-characteristics, pavlos-dataset-AOC} and used by many related studies~\cite{Gao-user-centric-DTN,multiple-offloading,CEDO}. Corollary~\ref{thm:corollary-pareto} follows directly from the expression of Lemma~\ref{thm:Pint-popularity}, and thus we omit the detailed proof.
\begin{corollary}\label{thm:corollary-pareto}
If $P_{p}^{req.}(n)$ is given by a Zipf distribution with a shape parameter $\alpha>0$, then $P_{p}(n)$ is given by a Zipf distribution with a shape parameter $\alpha +1$ (and vice versa).
\begin{equation*}
P_{p}^{req.}(n)\sim Zipf(\alpha) ~~\Leftrightarrow~~ P_{p}(n)\sim Zipf(\alpha+1) 
\end{equation*}
\end{corollary}
}

~\\
\noindent\underline{Answering \textbf{Q.2}}

The answer to question~Q.2 consists of two separate steps: (i) we calculate the number of holders \red{for the content of the random request}, and then (ii) we calculate how fast the requesting node can meet these holders. Towards answering (i), Lemma~\ref{thm:Pint-availability-generic} maps the popularity of the content involved in a random request (derived in Lemma~\ref{thm:Pint-popularity}) to the number of holders for this content. This number is a random variable dependent both on the popularity distribution $P_{p}(n)$, and on the availability function $g(m|n)$.
\begin{lemma}\label{thm:Pint-availability-generic}
The probability that a random request is for a content of availability equal to $m$ is given by
\begin{align}
 P_{a}^{req.}(m) = \frac{E_{p}[n\cdot g(m|n)]}{E_{p}[n]}\nonumber
\end{align}
\end{lemma}
\begin{proof}
The popularity of the content of a random request is given by $P_{p}^{req.}(n)$. Its availability can then be calculated by using the property of conditional expectation~\cite{RossProbModels}:
\begin{align}
 P_{a}^{req.}(m) = \sum_{n} P\lbrace N_{a} =m \vert N_{p} =n \rbrace\cdot P_{p}^{req.}(n)\nonumber
\end{align}
where $P_{p}^{req.}(n)$ is defined in Lemma~\ref{thm:Pint-popularity}. Substituting, from Def.~\ref{def:avail-popul-relation-generic} and Lemma~\ref{thm:Pint-popularity}, the above terms, we successively get
\begin{align}\label{eq:p-query-generic-3}
P_{a}^{req.}(m) &= \sum_{n} g(m|n)\cdot \frac{n}{E_{p}[n]}\cdot P_{p}(n)\nonumber\\
				 &= \frac{\sum_{n} g(m|n)\cdot n\cdot P_{p}(n)}{E_{p}[n]}= \frac{E_{p}[n\cdot g(m|n)]}{E_{p}[n]}\nonumber
\end{align}
which completes the proof.
\end{proof}

Having computed the statistics for the content availability, we can now calculate how fast the requesting node, say $j$, meets \textit{any} of the holders $i$ (i.e. nodes $i\in \mathcal{C}_{a}$). \red{As discussed in Section~\ref{sec:mobility-model}, we focus on the case of \textit{exponentially distributed inter-contact times}. Later, in Section~\ref{sec:analysis-extensions} we consider Pareto distributed inter-contact times as well.}

Let $T_{ij}$ denote the inter-contact times between node $j$ and a node $i\in\mathcal{C}_{a}$, and let  $T_{ij}$ be exponentially distributed with rate $\lambda_{ij}$. If we denote with $T_{\mathcal{M}}$ the first time until $j$ meets \textit{any} of the nodes $i\in\mathcal{C}_{a}$ (and, thus, accesses the content), then: 
\red{
\begin{equation*}
T_{\mathcal{M}} = \min_{i\in\mathcal{C}_{a}}\lbrace T_{ij}^{(r)}\rbrace
\end{equation*}
where $T_{ij}^{(r)}$ is the \textit{residual} inter-contact time between a node pair $\{i,j\}$. However, for the exponential distribution it holds that 
\begin{equation*}
T_{ij}\sim exponential(\lambda_{ij}) \Rightarrow T_{ij}^{(r)}\sim exponential(\lambda_{ij})
\end{equation*}
Therefore, $T_{\mathcal{M}}$ is distributed as a minimum of exponential random variables, and it follows that~\cite{RossProbModels}:}
\begin{equation}
T_{\mathcal{M}} \sim exp\left(\XM\right)~~\Leftrightarrow~~P\lbrace T_{\mathcal{M}}>t\rbrace  = e^{-\XM \cdot t}
\end{equation}
where
\begin{equation}\label{eq:XM-definition}
\XM = \sum_{i\in\mathcal{C}_{a}}\lambda_{ij}
\end{equation}

Clearly, knowing $\XM$ is needed to proceed with the desired metric derivation. Based on the preceding discussion, $\XM$ is a random variable that depends on: (i) the number of content holders $m$ (i.e. the cardinality of set $\mathcal{C}_{a}$ in Eq.(\ref{eq:XM-definition})), and (ii) the meeting rates with the holders. Applying Assumption~\ref{ass:traffic-mobility-independence}, it holds that, conditioning on $m$, $\XM$ (\eq{eq:XM-definition}) is a sum of $m$ i.i.d. random variables $\lambda_{ij}\sim f_{\lambda}(\lambda)$, i.e
\begin{equation}\label{eq:definition-fml}
\XM\sim f_{m\lambda}(x) = \left(f_{\lambda}\ast f_{\lambda} \cdots \ast f_{\lambda} \right)_{m},
\end{equation}
where $\ast$ denotes convolution, and mean value~\cite{RossProbModels}:
\begin{equation}\label{eq:mean-ml-x}
E[\XM|N_{a}=m] = E_{m\lambda}[x] =m\cdot \mu_{\lambda}
\end{equation}
\red{\textit{Remark}: In the remainder we use the subscript $m\lambda$ to denote an expectation over the distribution $f_{m\lambda}(x)$; the corresponding random variables are denoted as $x$.}

\subsection{Performance Metrics}\label{sec:performance-metrics}

We consider two main performance metrics: the \textit{average delay} and \textit{delivery probability}. Based on the analysis of Section~\ref{sec:preliminary-analysis}, we derive results under generic content traffic (i.e. $P_{p}(n)$ and $g(m|n)$) and mobility (i.e. $f_{\lambda}(\lambda)$) patterns.


\subsection*{\textbf{Content Access Delay}}

\begin{result}\label{result:ETm}
The expected content access delay can be computed with the expression
\begin{equation}
E[T_{\mathcal{M}}]= \frac{1}{E_{p}[n]}\cdot E_{p}\left[n\cdot\sum_{m} E_{m\lambda}\left[\frac{1}{x}\right] \cdot  g(m|n)\right]\nonumber
\end{equation}
\end{result}
\begin{proof}
The time $T_{\mathcal{M}}$ a node $j$ needs to access a content $\mathcal{M}$ is exponentially distributed with rate $X_{\mathcal{M}}$. However, $X_{\mathcal{M}}$ is a random variable itself, distributed with $f_{m\lambda}(x)$ (\eq{eq:definition-fml}). Thus, we can write for the expected content access delay:

\begin{footnotesize}
\begin{align}\label{eq:expected-Ta-generic}
&E[T_{\mathcal{M}}] 	= \sum_{m} E[T_{\mathcal{M}}\vert N_{a}=m] \cdot P_{a}^{req.}(m)\nonumber\\
					&= \sum_{m} \int E[T_{\mathcal{M}}\vert \XM=x,N_{a}=m]\cdot f_{m\lambda}(x)dx \cdot P_{a}^{req.}(m)\nonumber\\ 
					&= \sum_{m} \int \frac{1}{x}\cdot f_{m\lambda}(x)dx \cdot P_{a}^{req.}(m)
\end{align}
\end{footnotesize}
The last equality follows from the fact that the expectation of an exponential random variable with rate $x$ is $\frac{1}{x}$. 

Expressing the integral in \eq{eq:expected-Ta-generic} as an expectation over the $f_{m\lambda}(x)$ and substituting $P_{a}^{req.}(m)$ from Lemma~\ref{thm:Pint-availability-generic}, gives
\begin{align}\label{eq:expected-Ta-generic-1}
E[T_{\mathcal{M}}] 	&= \sum_{m} E_{m\lambda}\left[\frac{1}{x}\right] \cdot \frac{E_{p}[n\cdot g(m|n)]}{E_{p}[n]}\nonumber\\
					&= \frac{1}{E_{p}[n]}\cdot\sum_{m} E_{m\lambda}\left[\frac{1}{x}\right] \cdot E_{p}[n\cdot g(m|n)]
\end{align}
Rearranging the expectations and summation in \eq{eq:expected-Ta-generic-1} we get the expression of Result~\ref{result:ETm}.
\end{proof}

If the functions $f_{\lambda}(\lambda)$, $g(m|n)$ and $P_{p}(n)$ are known, the expected delay $E[T_{\mathcal{M}}]$ can be computed directly from Result~\ref{result:ETm}, as shown in the following example. 

\textit{\underline{Example Scenario:}} The contact rates ($f_{\lambda}$) follow a \textit{gamma distribution}, as suggested in~\cite{Passarella-aggregateIT}, with $\mu_{\lambda}$ and $CV_{\lambda}$.  Content popularity $P_{p}(n)$ is Pareto distributed, as observed in~\cite{youtube-traffic-from-edge, RSS-traffic-characteristics, pavlos-dataset-AOC}, with \textit{scale} and \textit{shape} parameters $n_{0}$ and $\alpha=2$, respectively. Finally, we consider a (deterministic) allocation of holders, $\rho(n)=c\cdot n$ (see Section~\ref{sec:content-traffic-model}). Then a closed form expression for $E[T_{M}]$ is given in the first row of Table~\ref{table:expressions-case-study}.

\begin{table}[!h]
\centering
\caption{Performance Metrics when $f_{\lambda}\sim Gamma$ with $\mu_{\lambda}, CV_{\lambda}$ and $P_{p}(n)\sim Pareto(n_{0},\alpha=2)$.}
\small
\begin{tabular}{|l|l|}
\hline
{}&{}\\
\hspace{-0.02\linewidth}
$\rho(n) = c\cdot n$ &
\hspace{-0.02\linewidth}
$E[T_{\mathcal{M}}]=\frac{1}{\mu_{\lambda}\cdot CV_{\lambda}^{2}} \left[\frac{c\cdot n_{0}}{CV_{\lambda}^{2}}\cdot\ln\left(\frac{1}{1-\frac{CV_{\lambda}^{2}}{c\cdot n_{0}}}\right)-1\right]$\\
{}&{}\\
\hline
\hspace{-0.02\linewidth}
$\rho(n) = c\cdot \ln(n)$ \hspace{-0.02\linewidth} &
\hspace{-0.02\linewidth}
$P\lbrace T_{\mathcal{M}}\leq TTL\rbrace = 1-\frac{1}{(1+\ln(\gamma))\cdot \gamma^{\ln(n_{0})}}$  \\
{}&
\multicolumn{1}{r|}{
\hspace{-0.02\linewidth}
where $\gamma = (1+\mu_{\lambda}\cdot CV_{\lambda}^{2}\cdot TTL)^{\frac{c}{CV_{\lambda}^{2}}}$
}\\
\hline
\end{tabular}
\normalsize
\label{table:expressions-case-study}
\end{table}

However, in a real implementation, it might not be always possible to know the \textit{exact} distributions of the contact rates ($f_{\lambda}$) and/or the availabilities ($g(m|n)$), needed to compute the expression of Result~\ref{result:ETm}. In the following theorem, we derive an expression for $E[T_{M}]$ that requires only the \textit{average statistics} (which are much easier to estimate or measure in a real scenario), namely (i) the mean value of the contact rates, $\mu_{\lambda}$, and (ii) the average availability for contents of a given popularity, $\overline{g}(n)$. 
\begin{theorem}\label{thm:lower-bound-ETm}
A lower bound for the expected content access delay is given by
\begin{equation}
E[T_{\mathcal{M}}]\geq \frac{1}{\mu_{\lambda}\cdot E_{p}[n]}\cdot E_{p}\left[\frac{n}{\overline{g}(n)}\right]\nonumber
\end{equation}
\end{theorem}
\begin{proof}
In Result~\ref{result:ETm} we can express $E_{m\lambda}\left[\frac{1}{x}\right]$ as $E_{m\lambda}[h(x)]$, where $h(x)=\frac{1}{x}$. Since $h(x)$ is a convex function, applying \emph{Jensen's inequality}, i.e. $h\left(E[x]\right)\leq E[h(x)]$, gives
\begin{equation}\label{eq:jensen-ml}
E_{m\lambda}\left[\frac{1}{x}\right]\geq \frac{1}{E_{m\lambda}[x]} = \frac{1}{m\cdot \mu_{\lambda}}
\end{equation}
where, in the equality, we used \eq{eq:mean-ml-x}.

Substituting \eq{eq:jensen-ml} in the expression of Result~\ref{result:ETm}, gives
\begin{align}\label{eq:expected-Ta-generic-4}
E[T_{\mathcal{M}}] 	&\geq	\frac{1}{ \mu_{\lambda}\cdot E_{p}[n]}\cdot E_{p}\left[n\cdot\sum_{m} \frac{1}{m} \cdot  g(m|n)\right]			
\end{align}
The sum in \eq{eq:expected-Ta-generic-4} is the expectation over $g(\cdot|n)$, i.e.
\begin{equation}
\sum_{m} \frac{1}{m} \cdot  g(m|n) = E_{g}\left[\frac{1}{m}\right]
\end{equation}
Applying, as before, Jensen's inequality, we get
\begin{equation}\label{eq:jensen-g}
\sum_{m} \frac{1}{m} \cdot  g(m|n) = E_{g}\left[\frac{1}{m}\right]\geq \frac{1}{E_{g}[m]}=\frac{1}{\overline{g}(n)}
\end{equation}
where we used for $E_{g}[m]$ the notation $\overline{g}(n)$.

Combining \eq{eq:jensen-g} and \eq{eq:expected-Ta-generic-4}, the expression of the theorem follows directly.
\end{proof}


\subsection*{\textbf{Content Access Probability}}

One often needs to also know the probability that a node can access a content by some deadline, i.e. $P\{T_{\mathcal{M}}\leq TTL\}$. E.g, a node might lose its interest in a content (e.g. news) after some time, or in an offloading scenario a node might decide to access a content directly to the base station. 
\begin{result}\label{result:P-Tm-TTL}
The probability a content to be accessed before a time $TTL$ can be computed with the expression
\begin{equation}
P\{T_{\mathcal{M}}\leq TTL\}= 1-\frac{E_{p}\left[n\cdot\sum_{m} E_{m\lambda}\left[e^{-x\cdot TTL}\right] g(m|n)\right]}{E_{p}[n]}\nonumber
\end{equation}
\end{result}
\begin{proof}
Conditioning on the values of $N_{a}$ and $\XM$, as in \eq{eq:expected-Ta-generic}, we can write:
\begin{small}
\begin{align}\label{eq:probability-Ta-generic}
&P\{T_{\mathcal{M}}\leq TTL\} 	=\nonumber\\
&= \sum_{m} \int P\{T_{\mathcal{M}}\leq TTL\vert \XM=x, N_{a}=m\}\cdot f_{m\lambda}(x)dx \cdot P_{a}^{req.}(m)\nonumber\\ 
					&= 1-\sum_{m} \int e^{-x\cdot TTL}\cdot f_{m\lambda}(x)dx \cdot P_{a}^{req.}(m)
\end{align}
\end{small}
where the last equality follows because $T_{\mathcal{M}}$ is exponentially distributed with rate $\XM=x$. %
After some similar steps as in Theorem~\ref{thm:lower-bound-ETm}, the final result follows.
\end{proof}
The expression of Result~\ref{result:P-Tm-TTL} for the previous example scenario, with a different allocation function $\rho(n) = c\cdot \ln(n)$, is given in the second row of Table~\ref{table:expressions-case-study}.

\begin{theorem}\label{thm:upper-bound-P-Tm-TTL}
An upper bound for the probability to access a content by a time $TTL$ is given by
\begin{equation}
P\{T_{\mathcal{M}}\leq TTL\}\leq 1-\frac{1}{E_{p}[n]}\cdot E_{p}\left[n\cdot e^{-\overline{g}(n)\cdot\mu_{\lambda}\cdot TTL}\right]\nonumber
\end{equation}
\end{theorem}
\begin{proof}
The bound follows easily by observing that $h(x)=e^{-x\cdot TTL}$ is a convex function, and applying \emph{Jensen's inequality} and the methodology of Theorem~\ref{thm:lower-bound-ETm}.
\end{proof}

\subsection*{\textbf{Tightness of bounds}}

\red{To derive simple expressions (bounds) that depend only on the average statistics $\mu_{\lambda}$ and $\overline{g}(n)$, and thus can be easily used in real scenarios (see, e.g., Section~\ref{sec:applications}), we applied twice Jensen's inequality. Although Jensen's inequality does not come with any quantitative guarantees for the tightness of a bound, in the following, we provide some intuition about how tight our expressions are expected to be in different scenarios.}

\red{Let us consider, for example, Theorem~\ref{thm:lower-bound-ETm} (similar arguments hold for Theorem~\ref{thm:upper-bound-P-Tm-TTL}). We first apply Jensen's inequality at \eq{eq:jensen-ml} for the expectation taken over node mobility, i.e.
\begin{equation*}
E_{m\lambda}\left[\frac{1}{x}\right]\geq \frac{1}{E_{m\lambda}[x]}
\end{equation*}
The same expectation, by applying the Delta method~\cite{Oehlert1992}, can be expressed as 
\begin{align*}
E_{m\lambda}\left[\frac{1}{x}\right] 
	&=  \frac{1}{E_{m\lambda}[x]} + \frac{E_{m\lambda}[(x-E_{m\lambda}[x])^2]}{(E_{m\lambda}[x])^{2}} + \cdots \\
	&=  \frac{1}{m\cdot \mu_{\lambda}}\cdot \left(1 + \frac{CV_{\lambda}^{2}}{m} + O\left(\frac{1}{m^2}\right)\right)
\end{align*}
As it can be seen in the above equation, the expectation $E_{m\lambda}\left[\frac{1}{x}\right]$ is equal to the lower bound, given by Jensen's inequality $\frac{1}{E_{m\lambda}[x]}$, plus a corrective term that decreases as (i) the heterogeneity of the mobility distribution (i.e. the variance $\sigma_{\lambda}$ and higher order moments) decreases, and (ii) the number of holders $m$ increases. Since $m$ takes every possible availability value, it follows that the tightness of the bound depends on the \textit{minimum availability} $m_{min}$ (and how probable this value is). For instance, if $m_{min}=1$ (and there is a high probability, i.e. $g(1|n)$, that this happens), the above bound probably might not be tight.}

\red{The second time we apply Jensen's inequality is at \eq{eq:jensen-g}, for the expectation over the availability distribution $g(m|n)$. Proceeding similarly, we can show that the bound
\[E_{g}\left[\frac{1}{m}\right]\geq\frac{1}{E_{g}[m]}\]
becomes tighter when the mass of the availability distribution $g(m|n)$ is concentrated around its mean value $E_{g}[m]$ (low heterogeneity). For example, for a deterministic $g(m|n)\rightarrow \rho(n)$, the bound is exact, whereas for a uniform distribution $g(m|n)=\frac{1}{C}, \forall m\in [1,W]$, the bound will become looser as $W$ increases.}

\red{Summarizing, the tightness of the bounds of Theorems~\ref{thm:lower-bound-ETm} and~\ref{thm:upper-bound-P-Tm-TTL} becomes higher as: 
\begin{itemize}
\item the heterogeneity of the mobility distribution $f_{\lambda}(\lambda)$ decreases
\item the minimum value of the availability, i.e. $m_{min} = min\{m:g(m|n)>0\}$, increases
\item the heterogeneity of the availability distribution $g(m|n)$ decreases
\end{itemize}
}

\subsection{Extensions}\label{sec:analysis-extensions}
In this section, we study how the results of Section~\ref{sec:performance-metrics} can be modified, when we remove the Assumptions~\ref{ass:availability-time-independence} and~\ref{ass:traffic-mobility-independence}. \red{Also, we provide the corresponding performance metric expressions for a Pareto distributed inter-contact times case.} We state here only the main findings and sketches of the proofs; the detailed proofs can be found in the Appendices.


\subsection*{\red{\textbf{Time-varying Availability: Multi-hop Content~Dissemination}}}
								
\red{In many protocols for opportunistic content-centric application proposed in literature, e.g.~\cite{Gao-user-centric-DTN,CEDO,offloading-wowmom11,contentplace}, the set of holders of a content might change over time or over the content distribution process, which is in contrast to Assumption~\ref{ass:availability-time-independence}. To this end, in this section, we study such cases of varying content availability. However, due to the numerous different approaches, each of them considering different ways of content dissemination (e.g. all nodes contribute to the content distribution~\cite{offloading-wowmom11}, or only selected nodes become holders~\cite{Gao-user-centric-DTN,contentplace}), a common methodology cannot be applied. Hence, we consider the following example scenario, and provide guidelines for analyzing further cases.} 

Let us assume a scenario where, initially, some nodes hold some \textit{content items} (e.g. data files), in which some other nodes are interested. This can be, for example, a content sharing scenario with contents being, e.g., some google maps. When a node interested in a content item, meets a holder and gets the content, it can hold it in its memory and act as a holder too. Specifically, we describe such scenarios as:

\begin{definition}\label{def:Availability Time Dependence}~\\
$-$~When a requester accesses a content, acts as a holder~for~it.\\
$-$ The \emph{initial} content popularity and availability patterns are given by $P_{p}(n)$ and $g(m|n)$.
\end{definition}

\red{In scenarios conforming to Def.~\ref{def:Availability Time Dependence}, an approximation\footnote{\red{The multi-hop delivery of a content, in combination with the mobility heterogeneity, does not allow the derivation of simple, closed-form expressions for exact predictions and bounds.}} for the expected content access delay $E[T_{\mathcal{M}}]$ is given by Result~\ref{R:ATD} (the detailed proof is given in Appendix~\ref{appendix:proof-R:ATD}).}
\begin{result}\label{R:ATD}
Under a time-varying availability scenario of Def.~\ref{def:Availability Time Dependence}, the expected content access delay is approximately given by
\begin{equation}
E[T_{\mathcal{M}}]=\frac{1}{\mu_{\lambda}\cdot E_{p}[n]}\cdot E_{p}\left[\ln\left(1+\frac{n}{\overline{g}(n)}\right)\right]\nonumber
\end{equation}
\end{result}
\textit{Sketch of proof:} Let us consider a content $\mathcal{M}$ of initial popularity $N_{p}(0)=n$ and availability $N_{a}(0)=m$. When the first requester accesses the content, the number of holders will increase to $m+1$ and the remaining requesters will be $n-1$. Building a Markov Chain as in Fig.~\ref{fig:markov-chain}, where each state denotes the number of holders, it can be shown for the expected delay of moving from state $m+k$ to state $m+k+1$, $k\in[0,n-1]$, that it holds $E[T_{k,k+1}] \approx \frac{1}{(m+k)\cdot (n-k)\cdot \mu_{\lambda}} $. Computing the times $E[T_{k,k+1}]$ and averaging over all the contents, gives the expected delay.
\begin{figure}
\centering
\includegraphics[width =1\linewidth]{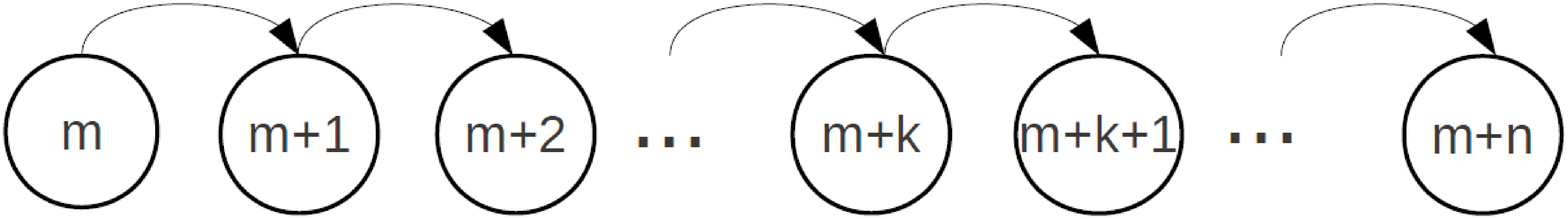}
\caption{Markov Chain for the dissemination of a content with initial popularity and availability $n$ and $m$, respectively.}
\label{fig:markov-chain}
\end{figure}

\red{Following the guidelines of the above methodology, further scenarios can be analysed as well. We provide here some examples (however, a detailed study is out of the scope of this paper):}

\red{\textit{Probabilistic cooperation.} A node receiving a content, might not be willing to cooperate and act as a holder for it (e.g. due to battery depletion, privacy concerns, etc.). To capture this, we can use the following model: a node acts as a holder for the content it receives with probability $p$. Then, at each content delivery, the number of holders increments with probability $p$ (and does not change with probability $1-p$). Building a similar Markov Chain as before, we can approximate the delay $E[T_{k,k+1}] \approx \frac{1}{(m+p\cdot k)\cdot (n-k)\cdot \mu_{\lambda}}$, where now $k\in[0,n-1]$ denotes the number of requesters that have been served. Hence, proceeding as in Appendix~\ref{appendix:proof-R:ATD}, we can get
\begin{equation}
E[T_{\mathcal{M}}]\approx\frac{1}{p\cdot \mu_{\lambda}\cdot E_{p}[n]}\cdot E_{p}\left[\ln\left(1+\frac{p\cdot n}{\overline{g}(n)}\right)\right]
\end{equation}}

\red{\textit{Limited spreading.} Let assume that the spreading of the content is limited to $L$ holders (i.e. only a limited number of $L$ \textit{new} holders is allowed), e.g., in order to reduce resource consumption in the network. In this case, the content availability will increase from $N_{a}(0) = m$ to its max value $N_{a} = m+L$ and after this point, it will not change. Then, the delay $\forall k\geq L$ (again $k$ refers to served requesters) will be given by $E[T_{k,k+1}] \approx \frac{1}{(m+L)\cdot (n-k)\cdot \mu_{\lambda}}$. Following the same steps as in  Appendix~\ref{appendix:proof-R:ATD}, we can get
\begin{equation}
~\hspace{-0.25cm}E[T_{\mathcal{M}}]\approx \frac{1}{\mu_{\lambda}\cdot E_{p}[n]}\cdot E_{p}\left[\frac{n-L}{\overline{g}(n)+L}+\ln\left(1+\frac{L}{\overline{g}(n)}\right)\right]
\end{equation}
}

\revisionRed{\underline{Remark:} It is possible in certain scenarios that content availability changes in various ways, sometimes not related (only) to the given content dissemination mechanism. For instance, holders may discard some contents due to limited resources, like full buffers, battery depletion, etc. An analysis as above could be applied for some of these cases as well (e.g., content discards could be modelled with a Markov Chain as in Fig.~\ref{fig:markov-chain}, where transitions to states with less holders are allowed). Due to space limitations, we defer the study of such interesting cases to future research.}

\subsection*{\revisionRed{\textbf{Time-varying Popularity}}}
\revisionRed{As discussed earlier, in the majority of the commonly considered applications/scenarios, users are not expected to change their interests in the time window of a content delivery; hence, content popularities do not change either. However, it is possible in certain cases that the popularity of a content might change over the (typical) time window of its delivery. In the following, we provide some initial analysis, as a first step towards analysing such cases\footnote{\revisionRed{We stress that a complete study of all the possible ways that the popularity patterns might change in an opportunistic content-centric scenario (and the respective analysis) is out of the scope of this paper.}}.}

\revisionRed{Let us assume a scenario where the initial requesters of a content start losing their interest with time. This is a common case among applications distributing news, trending video, etc. Since this loss of interest might appear in various ways (gradually, rapidly, etc.), which depend on the considered setting, we use the following generic way to model it.
\begin{definition}[Time-varying Popularity]
The probability a requester to have lost its interest by time $T$ is given by a distribution $P_{loss}\{t\leq T\}$.
\end{definition}
Under the above class of time-varying popularity cases, we can calculate the probability a content to be delivered to a requester by time $TTL$ (the expected delay is not a convenient performance metric in this case, since contents are never delivered to requesters that have lost their interest). If we denote this probability as $P_{delivery}$, it is easy to see that it holds
\begin{equation}
P_{delivery} = P\{T_{\mathcal{M}}\leq TTL \}\cdot \left(1-P_{loss}\{t\leq TTL\}\right)
\end{equation}
where $P\{T_{\mathcal{M}}\leq TTL \}$ is defined in Section~\ref{sec:performance-metrics}. $P_{delivery}$ can then be calculated straightforward from our results for $P\{T_{\mathcal{M}}\leq TTL \}$ (e.g., Result~\ref{result:P-Tm-TTL}) and the (known) distribution $P_{loss}\{t\leq T\}$.}

\revisionRed{\underline{Remark:} Further complexity can be added in the above model for users' loss of interest, like, heterogeneous distributions $P_{loss}$ (among users and/or contents), dependency between the loss of interest and the content or mobility characteristics, etc. The performance in such cases can be analysed similarly, based on our framework/methodology. }


\subsection*{\textbf{Mobility Dependent Allocation}}

\red{As discussed earlier (Section~\ref{sec:analysis}), who holds a content and who is interested in it, might be related to their mobility patterns, e.g. due to heterogeneous traffic patterns~\cite{pavlos-TMC-traffic} or a mobility-aware protocol~\cite{contentplace, multiple-offloading}. This can affect the performance in a positive or negative way, depending on the correlation between the mobility of holders and requesters. For instance, if a protocol selects as holders the nodes that meet more frequently the requesters (positive mobility correlation), then the performance is expected to be improved.}

\red{Due to the numerous different protocols and/or settings that might create such mobility correlations, we cannot analyze every single scenario separately. Hence, we choose to model the mobility dependence in a generic and probabilistic way. Then, to apply our results in a specific scenario, one needs only to make the correspondence between the mobility characteristics of the scenario and the model of Def.~\ref{def: Mobility Dependent Allocation} (e.g. following the guidelines of~\cite{pavlos-TMC-traffic}).}

\begin{definition}[Mobility Dependent Allocation]\label{def: Mobility Dependent Allocation}
The probability $\pi_{ij}$ that a node $i$ is a holder for a content in which a node $j$ is interested, is related to their contact rate $\lambda_{ij}$ such that $\pi_{ij} = \pi(\lambda_{ij})$, where $\pi(\cdot)$ is a function from $\mathbb{R}^{+}$ to $[0,1]$.
\end{definition}

\red{Based on the above definition, we can predict the performance of a content-centric application using Result~\ref{R:MDA}, which we prove in Appendix~\ref{appendix:R:MDA}}
\begin{result}\label{R:MDA}
 Under Def.~\ref{def: Mobility Dependent Allocation}, Theorems~\ref{thm:lower-bound-ETm} and~\ref{thm:upper-bound-P-Tm-TTL} and Result~\ref{R:ATD} hold if we replace $\mu_{\lambda}$ with $\mu_{\lambda}^{(\pi)}$, where
\begin{equation}
 \mu_{\lambda}^{(\pi)} = \frac{E_{\lambda}[\lambda\cdot \pi(\lambda)]}{E_{\lambda}[\pi(\lambda)]}\nonumber
\end{equation}
where $E_{\lambda}[\cdot]$ denotes an expectation taken over the contact rates distribution $f_{\lambda}(\lambda)$ (Assumption~\ref{ass:heterogeneous-mobility}).
\end{result}
\textit{Sketch of proof:} Since the requesters-holders contact rates are mobility dependent, the contact rates between them are not distributed with the contact rates distribution $f_{\lambda}(\lambda)$, but with a modified version of it, i.e. with a distribution:
\begin{equation}
 f_{\pi}(\lambda) = \frac{1}{E_{\lambda}[\pi(\lambda)]}\cdot \pi(\lambda)\cdot f_{\lambda}(\lambda)\nonumber
\end{equation}
Hence, \eq{eq:definition-fml} and \eq{eq:mean-ml-x} need to be modified as:
\begin{align}
&\XM\sim f_{m\pi}(x) = \left(f_{\pi}\ast f_{\pi} \cdots \ast f_{\pi} \right)_{m}\nonumber\\
&E[\XM|N_{a}=m] = E_{m\pi}[x] = m\cdot \frac{E_{\lambda}[\lambda\cdot \pi(\lambda)]}{E_{\lambda}[\pi(\lambda)]}= m\cdot \mu_{\lambda}^{(\pi)}\nonumber
\end{align}

\textit{\underline{Example Scenario:}} The holders of a content $\mathcal{M}$ are selected taking into account their contact rates with the requesters, as following: Each node $i$ (candidate to be a holder) is assigned a weight $w_{i} = \prod_{j\in\Cp}\lambda_{ij}$. Using such weights, the selection of holders that rarely meet the requesters is avoided. Then, each node is selected to be one of the $\Na$ holders with probability $p_{i} = \frac{w_{i}}{\sum_{i}w_{i}}$. With respect to Def.~\ref{def: Mobility Dependent Allocation}, it turns out that this mechanism is (approximately) described by $\pi(\lambda) = c\cdot \lambda$. Substituting $\pi(\lambda)$ in Result~\ref{R:MDA}, gives
\begin{equation}
 \mu_{\lambda}^{(\pi)} = \frac{E_{\lambda}[\lambda\cdot \pi(\lambda)]}{E_{\lambda}[\pi(\lambda)]}=\frac{E_{\lambda}[\lambda^{2}]}{E_{\lambda}[\lambda]} = \mu_{\lambda}\cdot(1+CV_{\lambda}^{2})
\end{equation}

\subsection*{\textbf{Pareto Inter-Contact Times}}

\red{We now proceed and demonstrate how our model can be extended to cases where inter-contact times between nodes are not exponentially distributed. Specifically, we consider inter-contact times following a Pareto distribution, which has been shown to fit some real traces~\cite{chaintreau-tmc}.} 

\red{Let us assume that inter-contact times between a node $j$ interested in a (random) content $\mathcal{M}$ and a node $i\in\mathcal{C}_{a}$ are Pareto distributed with \textit{shape} and \textit{scale} parameters $\alpha_{ij}+1$ (with $\alpha_{ij}>0$ when $E[T_{ij}]<+\infty$) and $t_{0}$, respectively\footnote{\red{We use the American Pareto (or Pareto Type II) distribution, which is supported for $t\geq0$~\cite{Boldrini-Residuals}. Moreover, for simplicity we assume a common scale parameter $t_{0}$ among all node pairs. Our results can be generalized for different scale parameters, e.g. $t_{0}^{(i,j)}$ for a pair $\{i,j\}$, however, this would increase the complexity of notation and expressions.}}:
\begin{equation*}
\textstyle T_{ij} \sim pareto(\alpha_{ij}+1, t_{0}) \Leftrightarrow P\lbrace T_{ij}>t\rbrace = \left(\frac{t_{0}}{t_{0}+t}\right)^{\alpha_{ij}+1}
\end{equation*}
Then, it follows that the residual inter-contact times will be also Pareto distributed, but with a decreased shape parameter~\cite{Boldrini-Residuals}, i.e.
\begin{equation*}
\textstyle T_{ij}^{(r)} \sim pareto(\alpha_{ij}, t_{0}) \Leftrightarrow P\lbrace T_{ij}^{(r)}>t\rbrace = \left(\frac{t_{0}}{t_{0}+t}\right)^{\alpha_{ij}}
\end{equation*}
and it can be shown for $T_{\mathcal{M}} = \min_{i\in\mathcal{C}_{a}}\lbrace T_{ij}^{(r)}\rbrace$ that (Appendix~\ref{appendix:min-pareto}):
\begin{equation*}
\textstyle T_{\mathcal{M}} \sim pareto(A_{\mathcal{M}}, t_{0}) \Leftrightarrow P\lbrace T_{\mathcal{M}}>t\rbrace = \left(\frac{t_{0}}{t_{0}+t}\right)^{A_{\mathcal{M}}}
\end{equation*}
where $A_{\mathcal{M}} = \sum_{i\in\mathcal{C}_{a}}\alpha_{ij}$.}

\emph{Remark:} \red{In this case the contact rates (Def.~\ref{ass:heterogeneous-mobility}) will be $\lambda_{ij} = \frac{1}{E[T_{ij}]} = \frac{\alpha_{ij}}{t_{0}}$, $\alpha_{ij}>0$}. However, for simplicity, we can use the parameters $\alpha_{ij}$ instead of the rates $\lambda_{ij}$, and, correspondingly, a distribution $f_{\alpha}(\alpha)$, instead of $f_{\lambda}(\lambda)$.

Hence, similarly to \eq{eq:definition-fml} and \eq{eq:mean-ml-x}, for Pareto intervals ($f_{a}(\alpha)$, $\mu_{\alpha}$), we can write:
\begin{equation}
A_{\mathcal{M}}\sim f_{m\alpha}(x) = \left(f_{\alpha}\ast\cdots \ast f_{\alpha} \right)_{m}, ~~E_{m\alpha}[x] = m\cdot \mu_{\alpha}\nonumber
\end{equation}

\red{Having calculated the above quantities, we can now proceed similarly to the exponential case (Section~\ref{sec:performance-metrics}) and derive the expressions for the performance metrics in the Pareto case (i.e. expressions corresponding to Results~\ref{result:ETm} and~\ref{result:P-Tm-TTL}, and Theorems~\ref{thm:lower-bound-ETm} and ~\ref{thm:upper-bound-P-Tm-TTL}). The expressions are given in Table~\ref{table:pareto-expressions} and the detailed derivations can be found in Appendix~\ref{appendix:pareto-expressions-proofs}.}

\begin{table*}
\centering
\caption{Performance metrics for Pareto distributed Inter-Contact times}
\label{table:pareto-expressions}
\begin{footnotesize}
\begin{tabular}{|l||c|c|}
\hline
{}&{Exact expressions}&{Bounds}\\
\hline
{}&{}&{}\\
{$E[T_{\mathcal{M}}]$}
&
$\displaystyle \frac{t_{0}}{E_{p}[n]}\cdot E_{p}\left[n\cdot \sum_{m} E_{m\alpha}\left[\frac{1}{x-1}\right]\cdot g(m|n)\right]$
&
$\displaystyle\frac{t_{0}}{E_{p}[n]}\cdot E_{p}\left[\frac{n}{\overline{g}(n)\cdot\mu_{\alpha}-1}\right]$
\\
{}&{}&{}\\
\hline
{}&{}&{}\\
{$P\{T_{\mathcal{M}}\leq TTL\}$}
&
$\displaystyle 1-\frac{1}{E_{p}[n]}\cdot E_{p}\left[n\cdot\sum_{m} E_{m\alpha}\left[\left(\frac{t_{0}}{t_{0}+TTL}\right)^{x}\right]\cdot g(m|n)\right]$
& 
$\displaystyle 1-\frac{1}{E_{p}[n]}\cdot E_{p}\left[n\cdot \left(\frac{t_{0}}{t_{0}+TTL}\right)^{\overline{g}(n)\cdot\mu_{\alpha}}\right]$\\
{}&{}&{}\\
\hline
\end{tabular}
\end{footnotesize}
\end{table*}

\subsection{Model Validation}\label{sec:synthetic-simulations}

As a first validation step, we compare our theoretical predictions to synthetic simulation scenarios conforming to the models of Section~\ref{sec:network-model}, in order to consider (a) various mobility and content traffic patterns, and (b) large networks. 

\textit{\textbf{Simulation Scenarios:}} We assign to each pair $\{i,j\}$ a contact rate $\lambda_{ij}$, which we draw randomly from a distribution $f_{\lambda}(\lambda)$, and create a sequence of contact events (Poisson process with rate $\lambda_{ij}$). Then, we create $M$ contents and assign to each of them a popularity value ($N_{p}$), drawn from the distribution $P_{p}(n)$. According to the given function $g(m|n)$, we assign the availability values ($N_{a}$). Finally, for each content $\mathcal{M}$, we randomly choose the $\Np$ nodes that are interested in it and its $\Na$ holders. 

\begin{figure}
\centering
\subfigure[{$E[T_{\mathcal{M}}]$}]{\includegraphics[width=0.49\linewidth]{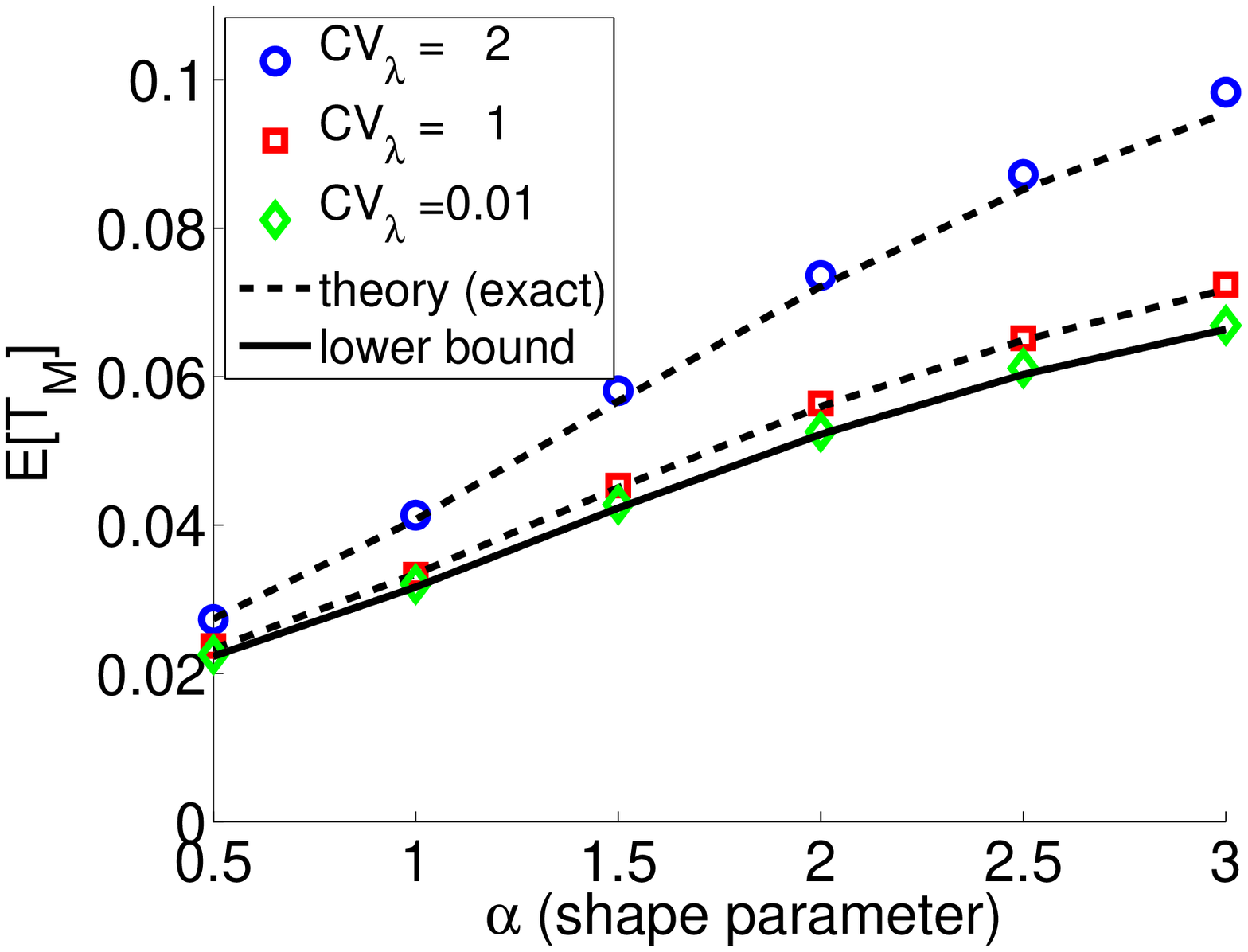}\label{fig:synth-ETm-gLinear}}
\subfigure[$P\{T_{\mathcal{M}}\leq TTL\}$]{\includegraphics[width=0.49\linewidth]{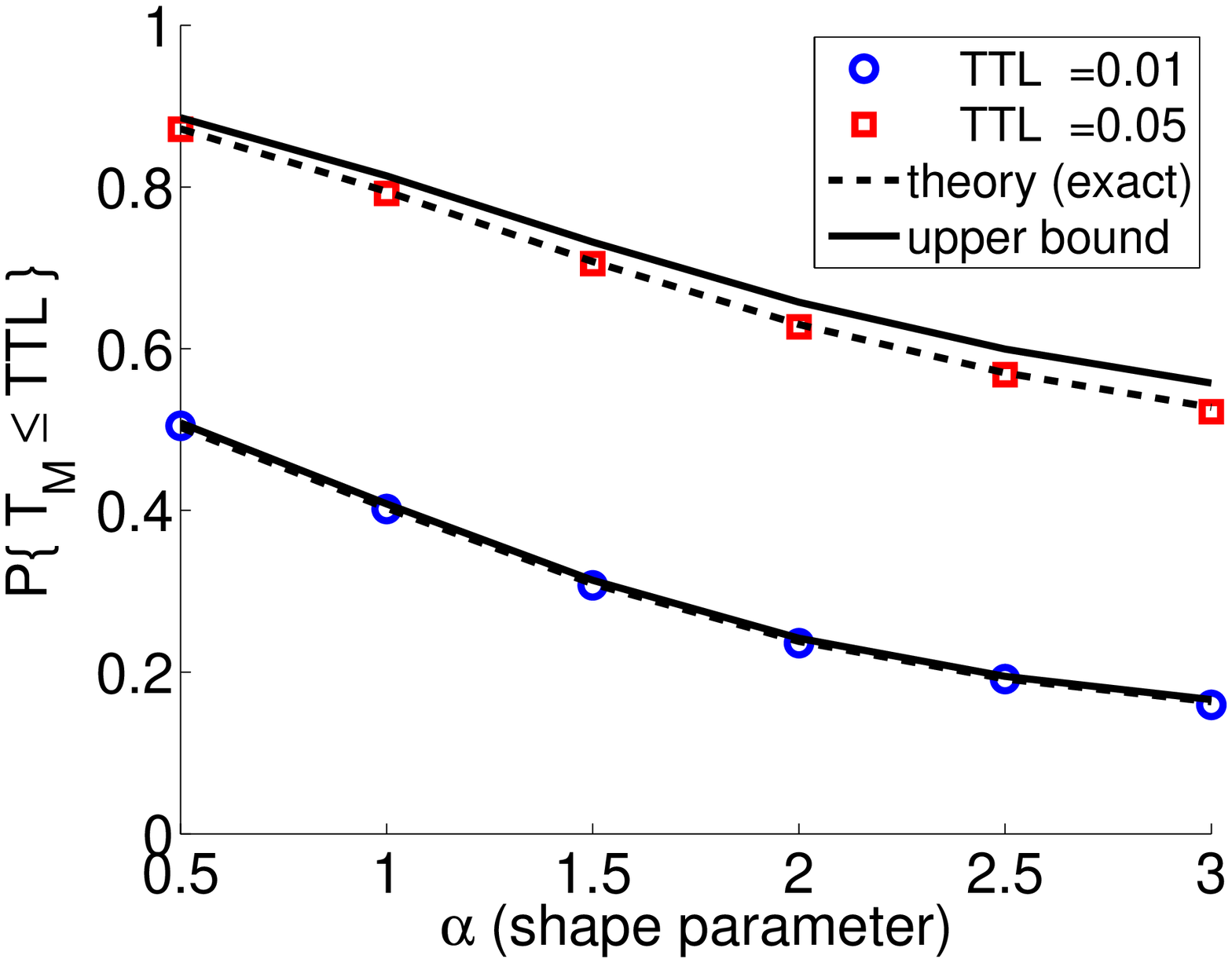}\label{fig:synth-P-gLinear-cv2}}
\caption{(a) $E[T_{\mathcal{M}}]$ and (b) $P\{T_{\mathcal{M}}\leq TTL\}$ in scenarios with varying content popularity ($\alpha$: shape parameter).
}
\label{fig:synth-gLinear}
\end{figure}

\textit{\textbf{Mobility / Popularity patterns:}} In most of the scenarios we present, we use the Gamma distribution for the contact rates (i.e. $f_{\lambda}(\lambda)$), since it has been shown to match well characteristics of real contact patterns~\cite{Passarella-aggregateIT}. Also, content popularity in mobile social networks has been shown to follow a power-law distribution, e.g.~\cite{youtube-traffic-from-edge, RSS-traffic-characteristics, pavlos-dataset-AOC}. Therefore, we select $P_{p}(n)$ to follow Discrete (Bounded) Pareto or Zipf distributions, similarly to the majority of related works~\cite{Gao-user-centric-DTN,multiple-offloading,CEDO}.

In Fig.~\ref{fig:synth-gLinear} we present the simulation results, along with our theoretical predictions, in scenarios of $N=10000$ nodes with varying mobility and content popularity patterns. The mean contact rate is $\mu_{\lambda}=1$ and content popularity follows a Bounded Pareto distribution with shape parameter (i.e. exponent) $\alpha$ and $n\in[50, 1000]$. The availability function is $\rho(n)=0.2\cdot n$ (i.e. deterministic). An almost perfect match between simulation results (markers) and the theoretical predictions (dashed lines) of Results~\ref{result:ETm} and~\ref{result:P-Tm-TTL} can be observed. In Fig.~\ref{fig:synth-ETm-gLinear}, the lower bound (continuous line) of Theorem~\ref{thm:lower-bound-ETm} is very tight for low mobility (i.e. small $CV_{\lambda}$) and/or content popularity (i.e. small $\alpha$) heterogeneity\red{, confirming thus the discussion of Section~\ref{sec:performance-metrics} for the bound tightness}. For the delivery probability $P\lbrace T_{\mathcal{M}}\leq TTL\rbrace$ (Fig.~\ref{fig:synth-P-gLinear-cv2}), we present the results for two different values of $TTL$ in scenarios with $CV_{\lambda}=2$ \red{(i.e. the most heterogeneous scenario)}. Here, the upper bound (continuous line) of Theorem~\ref{thm:upper-bound-P-Tm-TTL} is very close to the simulation results, despite the very heterogeneous mobility.

In Table~\ref{Table:synthetic-simulation-results-binomial} we present results of the above scenarios, where the availability - popularity correlation is not deterministic. We assume that $g(m|n)$ follows a binomial distribution with mean $\overline{g}(n) = 0.2\cdot n$. \red{The binomial distribution introduces a randomness that can be interpreted as noise in a system's availability estimation algorithm, differences in node behaviors (e.g. a node having a resource, shares it with probability $p$), etc.} It can be seen that the bounds are tight in most of the scenarios, though (as expected - cf. Section~\ref{sec:performance-metrics}) less tight than in the deterministic $g(m|n)$ case (i.e. $\rho(n)$).

\begin{table}[!h]
\centering
\caption{Simulation results for scenarios where $g(m|n)\sim Binomial$ with $\overline{g}(n) = 0.2\cdot n$, and $TTL=0.05$.}
\begin{footnotesize}
 \begin{tabular}{|l|cccc|}
\hline
{$E[T_{\mathcal{M}}]$ ($x 10^{3}$)}
			&   $\alpha = 0.5$  &   $\alpha = 1$  &   $\alpha = 2$  &   $\alpha = 3$    \\
\hline
lower bound             &   22.3    	    &	31.6	      &		52.2  	& 	 66.4     \\
simulation ($CV_{\lambda}=0.5$) &   23.9    &	34.8          &		57.3  	&  	 75.0     \\
simulation ($CV_{\lambda}=1$)   &   25.0    &	36.2          &		61.9  	&  	 81.4     \\
\hline
{$P\{T_{\mathcal{M}}\leq TTL\}$}
			&   $\alpha = 0.5$  &   $\alpha = 1$  &   $\alpha = 2$  &   $\alpha = 3$    \\
\hline
upper bound             &   0.89    	    &	0.81	      &		0.66  	& 	 0.56     \\
simulation ($CV_{\lambda}=2$)&   0.87 	    &	0.79          &		0.62  	&  	 0.52     \\
\hline
 \end{tabular}
\label{Table:synthetic-simulation-results-binomial}
\end{footnotesize}
\end{table}

\red{Finally, Table~\ref{Table:relative-error-vsNetworkSize} shows the accuracy of our results in smaller network size scenarios with $f_{\lambda}\sim Gamma(\mu_{\lambda}=1,CV_{\lambda}=1)$ and $P_{p}(n)\sim BoundedPareto(\alpha=2)$ and $n\in[50,500]$. We present the relative errors between the simulation values and the predictions of Results~\ref{result:ETm} and~\ref{result:P-Tm-TTL}. It can be seen that as the network size decreases the error increases; however, the accuracy is significant for all cases (max error $\approx 5\%$).
}
\begin{table}[!h]
\centering
\caption{Relative error between simulation results and Results~\ref{result:ETm} and~\ref{result:P-Tm-TTL} for various network size scenarios.}
 \begin{tabular}{|l|cccc|}
\hline
$N$		& 500 	&1000 	&1500	&2000\\
\hline
rel. error, $E[T_{\mathcal{M}}]$ &  4.98\% & 1.79\% & 1.25\% & 1.08\% \\
rel. error, $P\{T_{\mathcal{M}}\leq TTL\}$ &  5.24\% & 1.27\% & 1.03\% & 0.77\%\\
\hline
 \end{tabular}
\label{Table:relative-error-vsNetworkSize}
\end{table}

\red{We, now, proceed in the validation of the extensions of our basic results presented in Section~\ref{sec:analysis-extensions}.} First, in Fig.~\ref{fig:synth-gLinear-extensions-Time} we compare Result~\ref{R:ATD} with simulations on scenarios conforming to Def.~\ref{def:Availability Time Dependence}: $P_{p}(n)$ is a \textit{Bounded Pareto} distribution with $\alpha=2$, and $f_{\lambda}(\lambda)\sim$ \textit{Pareto}, \red{which can be a reasonable choice for opportunistic networks~\cite{Passarella-aggregateIT}\footnote{\red{We use a Pareto, instead of a Gamma, distribution, in order to be able to achieve high $CV_{\lambda}$ values without having to decrease the $min\{\lambda_{ij}\}$ value.}}. Since Result~\ref{R:ATD} is based on an approximation that is more accurate for less heterogeneous mobility patterns, we compare our predictions with simulations in scenarios with varying $CV_{\lambda}$.} It can be seen that our theoretical prediction (approximation) achieves good accuracy even in these very heterogeneous mobility scenarios.

Results for scenarios with mobility-dependent availability (Def.~\ref{def: Mobility Dependent Allocation}) are presented in Fig.~\ref{fig:synth-gLinear-extensions-Mobility}. $P_{p}(n)$ is selected as before and $f_{\lambda}(\lambda)\sim$ \textit{Gamma} with $\mu_{\lambda}=1,CV_{\lambda}=0.5$. The allocation of holders is made as in the example scenario of Section~\ref{sec:analysis-extensions}. The upper bounds of Result~\ref{R:MDA} are tight in all scenarios, similarly to the case without mobility dependence (Fig.~\ref{fig:synth-P-gLinear-cv2}).

\begin{figure}
\centering
\subfigure[{$E[T_{\mathcal{M}}]$}]{\includegraphics[width=0.49\linewidth]{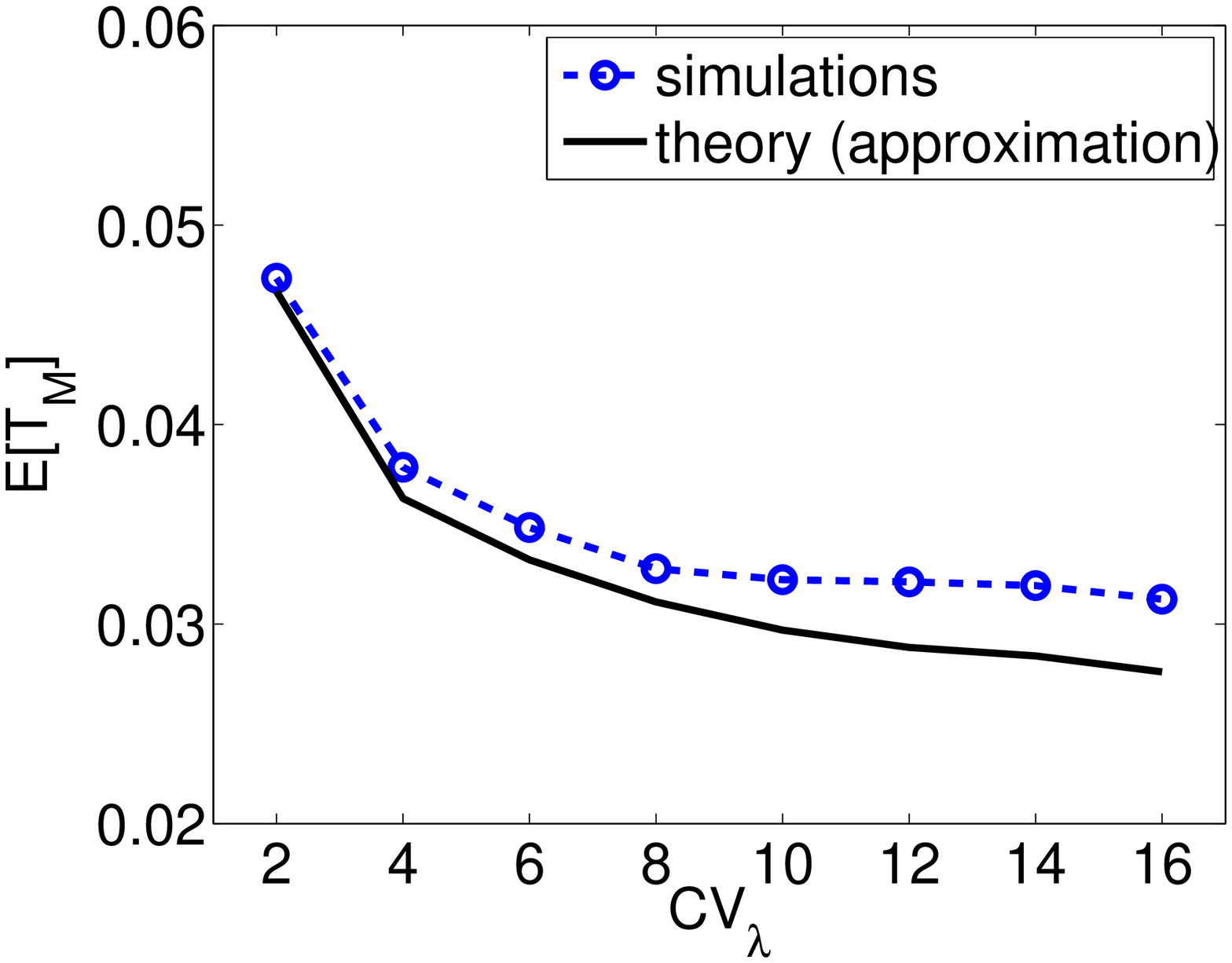}\label{fig:synth-gLinear-extensions-Time}}
\subfigure[$P\{T_{\mathcal{M}}\leq TTL\}$]{\includegraphics[width=0.49\linewidth]{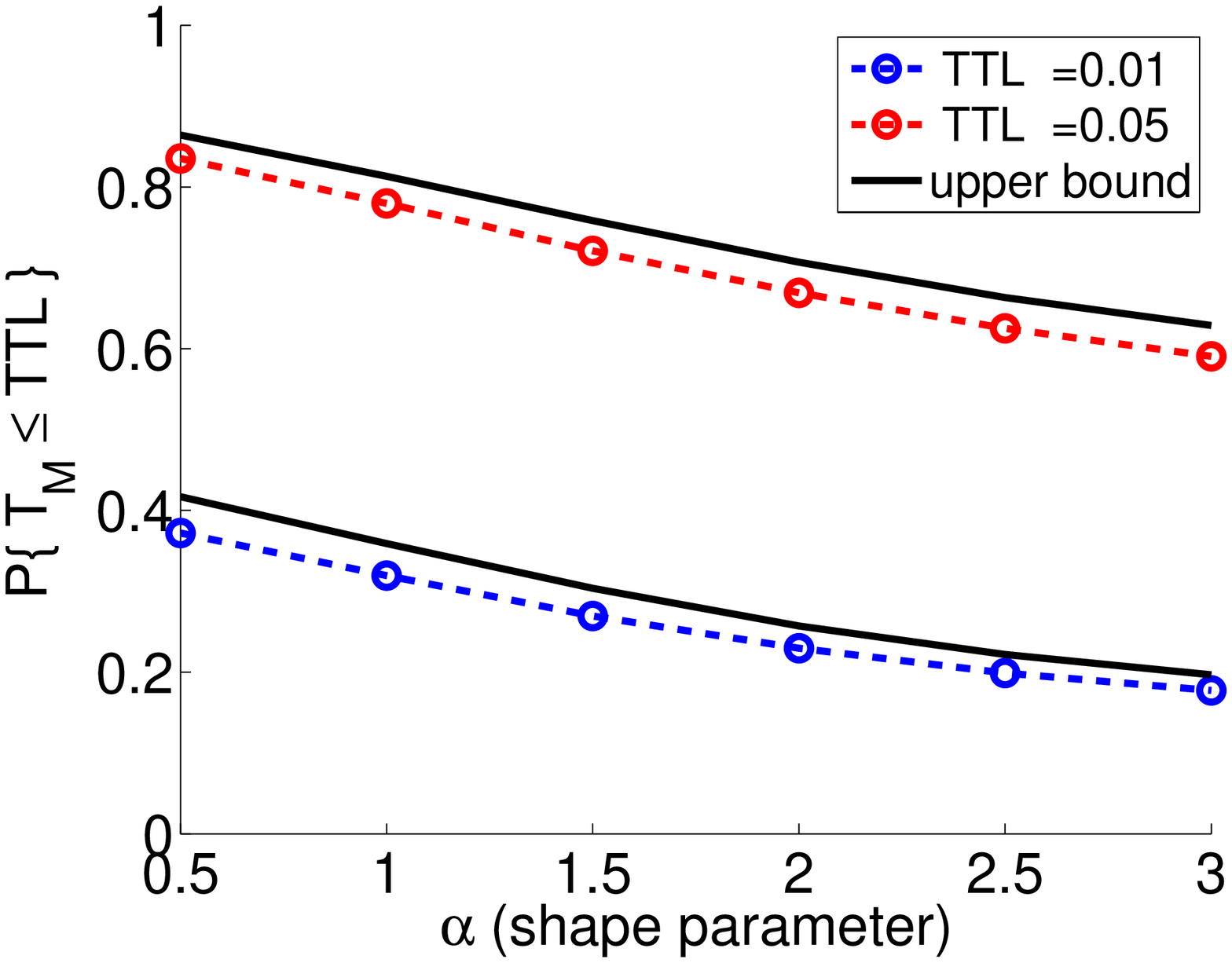}\label{fig:synth-gLinear-extensions-Mobility}}
\caption{(a) $E[T_{\mathcal{M}}]$ in scenarios under Def.~\ref{def:Availability Time Dependence} and (b) $P\{T_{\mathcal{M}}\leq TTL\}$ in scenarios under Def.~\ref{def: Mobility Dependent Allocation}. $\rho(n) = 0.2\cdot n$.
}
\label{fig:synth-gLinear-extensions}
\end{figure}

\red{Finally, we simulate scenarios with Pareto distributed inter-contact times, as assumed in Section~\ref{sec:analysis-extensions}. We consider two scenarios with $t_{0}=1$ and shape parameters $\alpha_{ij}$ uniformly distributed in the intervals $[1.5, 4]$ and $[1.5, 6]$, respectively. We present the simulation results, along with the theoretical bounds of Table~\ref{table:pareto-expressions} in Fig.~\ref{fig:paretoICT}. As it can be seen the bounds are tight in all cases. The accuracy of the exact predictions of Table~\ref{table:pareto-expressions} is significant as well (it is not shown in the figure).}
\begin{figure}
\centering
\subfigure[{$E[T_{\mathcal{M}}]$}]{\includegraphics[width=0.49\linewidth]{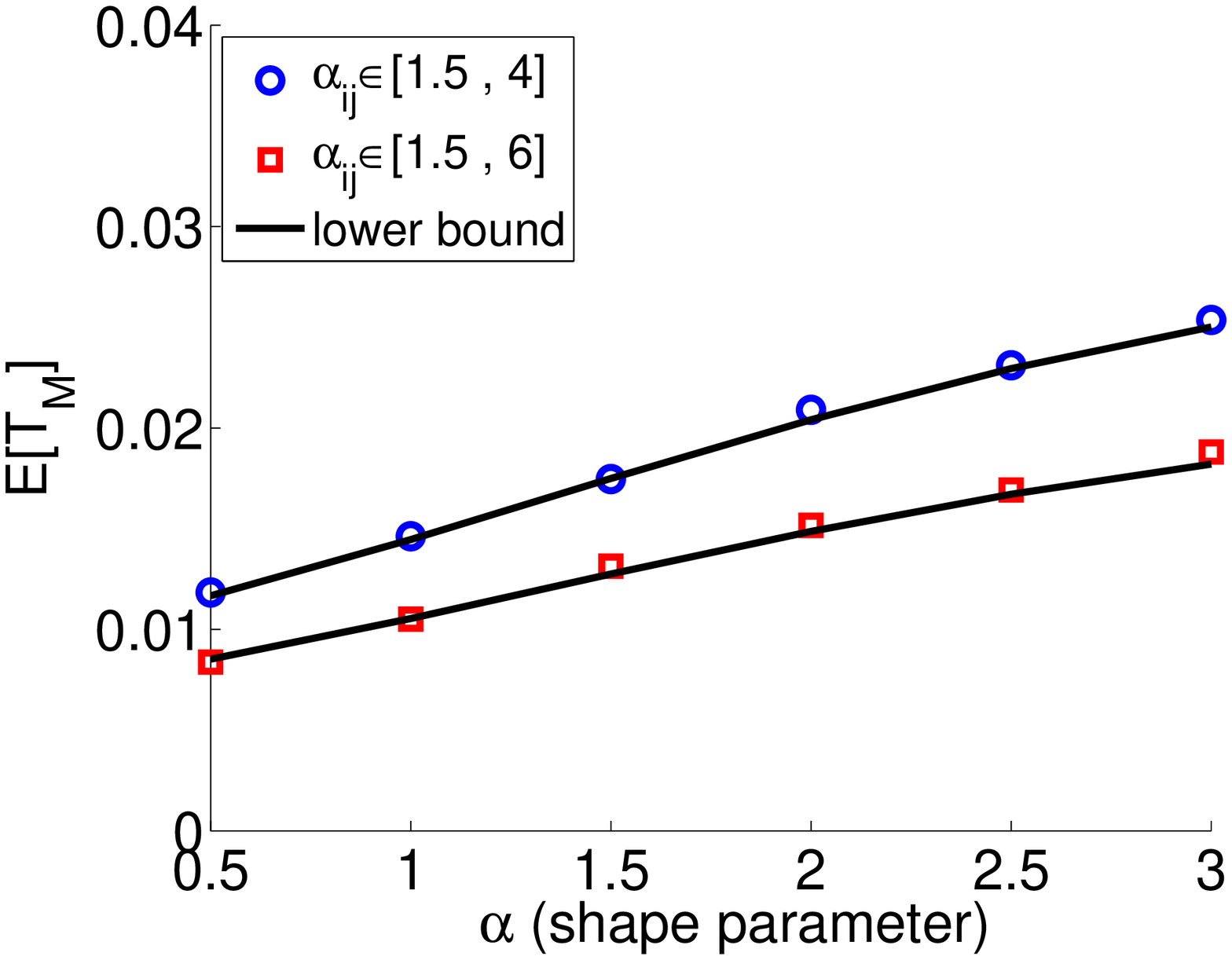}\label{fig:paretoICT-delay}}
\subfigure[$P\{T_{\mathcal{M}}\leq TTL\}$]{\includegraphics[width=0.49\linewidth]{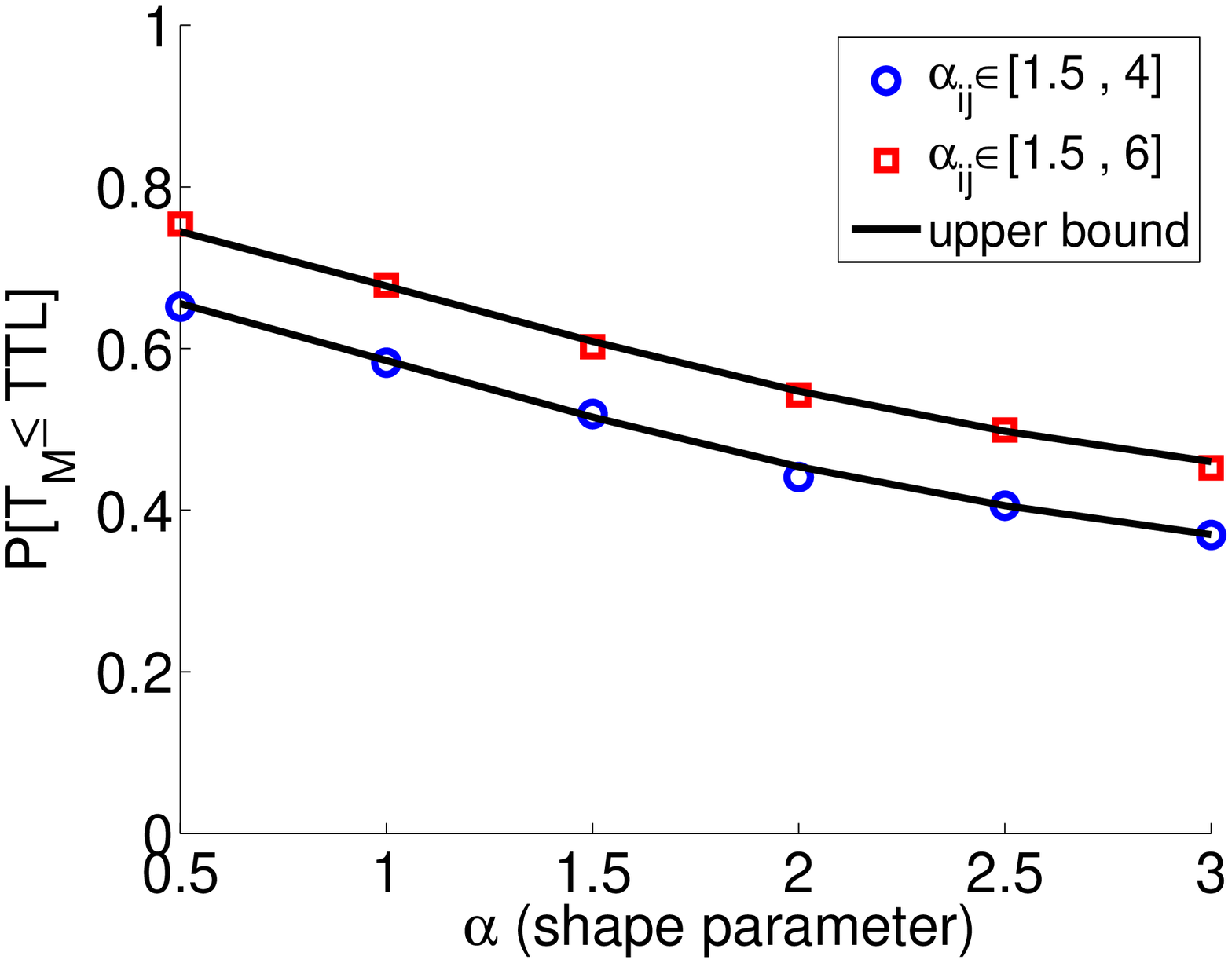}\label{fig:paretoICT-probability}}
\caption{Simulation results and theoretical bounds (Table~\ref{table:pareto-expressions}) (a) $E[T_{\mathcal{M}}]$ and (b) $P\{T_{\mathcal{M}}\leq TTL\}$ in scenarios with Pareto Inter-Contact times. $P_{p}\sim Pareto(\alpha)$ with $n\in[50,100]$, and $\rho(n) = 0.2\cdot n$.}
\label{fig:paretoICT}
\end{figure}

\section{Case Study: Mobile Data Offloading}\label{sec:applications}
The results of Section~\ref{sec:analysis} can be used to predict the performance of a given content allocation policy or content-sharing scheme. In this section, we show how these results could be also used to design / optimize policies. We focus on an application that has recently attracted attention, that of \textit{mobile data offloading} using opportunistic networking~\cite{offloading-wowmom11,Hui-Offloading, multiple-offloading}. Nevertheless, the same methodology applies for a range of other applications where the number of content/service providers must be chosen.    

\red{In a mobile data offloading scenario, the goal of the cellular network provider is to reduce the traffic served by the infrastructure. To achieve this, the cellular network, instead of transmitting separately a content to every node interested in it, distributes content copies only to some of the interested nodes (holders). The remaining (interested) nodes must then retrieve the content from the designated holders during direct encounters. In some cases, an additional \textit{QoS} constraint might exist: if the delay to access a content exceeds a $TTL$, a requesting node will download it from the infrastructure~\cite{offloading-wowmom11,Hui-Offloading, multiple-offloading}.}

\red{A tradeoff is involved between the amount of traffic offloaded and the average delay for non-holders: transmitting the content to less holders, increases the traffic that is offloaded, but also increases the time needed by a node to encounter a holder and get the content. Similar tradeoffs (between the amount of offloaded traffic and $P\{T_{\mathcal{M}}\leq TTL\}$) appear in the QoS case as well. Hence, the cellular provider has to find a point in this tradeoff -by selecting the set (or number) of the holders- that satisfies both its need to alleviate the infrastructure and the users' demands (e.g. low delivery delay). Moreover, when many different contents have to be offloaded, the number of holders that can be allocated for each of them might be constrained. The reason for this can be related to the fact that (a) the total number of cellular transmissions (which is equal to the total number of holders for all the messages) is limited due to the congestion of the wireless interface, and/or (b) nodes have limited resources, like energy or buffer size, and thus they cannot store and forward many contents.}

\red{Algorithm~\ref{alg:offloading} summarizes the main functions of a mobile data offloading system as described above.}

\red{\textit{Remark:} Here, we would like to remind the reader that we study mobile data offloading as an example showing how our model and analysis can be applied; describing in detail how to design a system implementing Algorithm~\ref{alg:offloading} is out of the scope of the paper.}
\begin{algorithm}[h]
\caption{Mobile Data Offloading}
\begin{algorithmic}[1]
\State \textbf{Input:} \textit{contents}, \textit{mobility}, \textit{constraints}
\State $\mathcal{H}\gets $ select\_holders(\textit{contents}, \textit{mobility}, \textit{constraints});
\State send\_copies\_to\_holders(\textit{contents},$\mathcal{H}$);
\State opportunistic\_content\_delivery();
\If {$QoS$}
    \State transmit\_undelivered\_contents(\textit{contents}, $TTL$);
\EndIf
\end{algorithmic}
\label{alg:offloading}
\end{algorithm}

\red{The input needed by the cellular network consists of: (i) The set of the contents $\textbf{M}$, which is already known, since nodes request the contents from the cellular network\footnote{\red{In an alternative scenario it could happen that the contents are not known \textit{a priori} and the cellular network pre-caches some contents to avoid future requests. In this case, although the exact set of the nodes interested in each content is not known, estimations (e.g., based on regular patterns, past data, or prediction methods~\cite{thesis-popularity-prediction}) about the intensity of requests (i.e. number of contents $M$) and the popularity would suffice (as can be seen, e.g., in Result~\ref{result:optimal-delay}).}}; (ii) the mobility patterns of the nodes (or only some average statistics of them, as we show later), which can be estimated, e.g. by past data~\cite{Hui-Offloading}; and (iii) the constraint on the total number of holders, which can be calculated by directly obtaining information from the nodes or it is a parameter controlled by the system, in the cases that it corresponds to node resources or the max number of transmissions, respectively.} 

\red{The next step (line 2 in Algorithm~\ref{alg:offloading}), which is the main focus of this section, is to \textit{choose the set of holders for each content}. The cellular network provider tries to find the allocation that optimizes a performance metric, under the given set of contents, the node mobility and the popularity distribution. Then, the selected holders receive the contents from the cellular network (line 3) and forward them to other interested nodes they encounter (line 4). Finally, if a $QoS$ constraint exists, the nodes that have not received the contents by time $TTL$, get them from the cellular network (lines 5-7).} 

\red{As said earlier, in this section we try to optimally allocate holders for a mobile data offloading scenario.} We study cases with and without $QoS$ constraints in Sections~\ref{sec:case-study-ETm} and~\ref{sec:case-study-Roff}, respectively. For simplicity, we use the expressions of Theorems~\ref{thm:lower-bound-ETm} and~\ref{thm:upper-bound-P-Tm-TTL} as approximations for $E[T_{\mathcal{M}}]$ and $P\{T_{\mathcal{M}}\leq TTL\}$. Since, these expressions imply that (a) the exact mobility patterns are not known (i.e. only $\mu_{\lambda}$ is needed) and (b) contents with the same popularity are equivalent, our goal is to select the number of holders for each content with a given popularity. In other words, we try to find the optimal allocation function $g(m|n)$.

\subsection{Case 1: no QoS constraints}\label{sec:case-study-ETm}
When no $QoS$ constraints exist, the cellular operator decides the maximum amount of traffic that it wishes to serve directly over the infrastructure. Under this constraint, which can be translated as a constraint on the number of holders for different contents, the objective is to minimize the expected delay $E[T_{\mathcal{M}}]$. The following result (proved in Appendix~\ref{appendix:proof-result:optimal-delay}), formalizes this optimization problem and provides with the optimal solution for $g(m|n)$.

\begin{result}\label{result:optimal-delay}
The minimum expected content access delay, under the constraint of an average number of $c_{\mathcal{M}}$ copies per content, i.e.:
\begin{equation}
 \min\{E[T_{\mathcal{M}}]\}~~~~s.t.~~~\sum_{\mathcal{M}} \Na= M\cdot c_{\mathcal{M}}~,~\Na\geq0\nonumber
\end{equation}
can be achieved when the allocation function, $g(m|n)$, is deterministic and equal to
\begin{equation}\label{eq:optimal-g-delay}
 \rho^{*}(n) = \frac{c_{\mathcal{M}}}{E_{p}[\sqrt{n}]}\cdot \sqrt{n}\nonumber
\end{equation}
\end{result}

Result~\ref{result:optimal-delay} is a generic result, since it holds under \textit{any} content popularity pattern. We also note that an allocation policy of $\rho(n) \propto \sqrt{n}$ has also been shown to achieve optimal results in (conventional) peer-to-peer networks~\cite{replication-p2p}. This is an interesting finding, given the inherent differences between the two settings (e.g. node mobility).

Finally, our result is also consistent in scenarios with \textit{mobility dependent holders allocation}. For example, after choosing the number of copies for a content (Result~\ref{result:optimal-delay}), the selection of holders can be made, taking into account mobility utility metrics, e.g. meeting frequency~\cite{contentplace} or node centrality~\cite{Gao-user-centric-DTN}.

\subsection{Case 2: QoS constraints}\label{sec:case-study-Roff}
In cases where a maximum delay $TTL$ is required, the objective is to minimize the traffic load served by the infrastructure. The metric used in related work, e.g.~\cite{multiple-offloading}, is the data offloading ratio, $R_{off.}$, which is defined as the percentage of content requests that are served by nodes. Since requests are served by the infrastructure only after the time $TTL$ elapses, it follows that in our framework: $ R_{off.} = P\lbrace T_{\mathcal{M}} \leq TTL\rbrace$.

Hence the optimization problem is equivalent to
\begin{equation}
 \max P\lbrace T_{\mathcal{M}} \leq TTL\rbrace~~s.t.~~\sum_{\mathcal{M}} \Na= M\cdot c_{\mathcal{M}},~\Na\geq0\nonumber
\end{equation}


Proceeding similarly to the proof of Result~\ref{result:optimal-delay} (see Appendix~\ref{appendix:proof-result:optimal-delay}), the above optimization problem becomes:
\begin{equation}
 \min_{\rho(n)}\lbrace E_{p}\left[n\cdot e^{-\rho(n)\cdot\mu_{\lambda}\cdot TTL}\right]\rbrace ~~~~s.t.~~~E_{p}[\rho(n)]= c_{\mathcal{M}}
\end{equation}
with $\rho(n)\geq0$, or, equivalently (by expressing the expectation as a sum and denoting $\rho_{n} = \rho(n)$):
\begin{align}\label{eq:optimization-problem-Roff-final-expression}
& \textstyle\min_{\left\{\rho_{1},...,{\rho_{n}}\right\}}~~\lbrace\sum_{n}n\cdot e^{-\rho_{n}\cdot\mu_{\lambda}\cdot TTL}\cdot P_{p}(n)\rbrace\nonumber\\ 
&\textstyle~~~~~~~s.t.~~~\sum_{n}\rho_{n}\cdot P_{p}(n)= c_{\mathcal{M}}~,~~\rho_{n}\geq0
\end{align}
The optimization problem of \eq{eq:optimization-problem-Roff-final-expression} is convex. Although a closed form solution, as in Result~\ref{result:optimal-delay}, cannot be derived, it can be solved numerically, using well known methods.

\subsection{Performance Evaluation}\label{sec:offloading-evaluation}
To investigate whether the policies suggested as optimal by our theory indeed perform better, we conducted simulations on various synthetic scenarios and on traces of real networks, where node mobility patterns usually involve much more complex characteristics than our model (Assumption~\ref{ass:heterogeneous-mobility}).

The results in the majority of scenarios considered have been encouragingly consistent with our theoretical predictions. Hence, we only present here a small, representative sample. Specifically, we consider the following traces coming from state-of-the-art mobility models or collected in experiments.

\noindent\textit{\textbf{TVCM}} mobility model~\cite{tvcm}: Scenario with $100$ nodes divided in $4$ communities of unequal size. Nodes move mainly inside their community and leave it for a few short periods.\\
\textit{\textbf{SLAW}} mobility model~\cite{slaw}: Network with $200$ nodes moving in a square area of $2000m$ (the other parameters are set as in the source code provided in~\cite{slaw}).\\
\textit{\textbf{Cabspotting}} trace~\cite{cabspotting-trace}: GPS coordinates from $536$ taxi cabs collected over 30 days in San Francisco. A range of $100m$ is assumed.\\
\textit{\textbf{Infocom}} trace~\cite{Infocom-trace}: Bluetooth sightings of $98$ mobile and static nodes (iMotes) collected in an experiment during Infocom 2006.

\subsubsection{Case 1: no $QoS$ constraints}
In each scenario, we compare different allocation functions $\rho(n) = c_{k}\cdot n^{k}$, where $c_{k} = \frac{c_{\mathcal{M}}}{E_{p}[n^{k}]}$ is a normalization factor such that the constraint $E_{p}[\rho(n)]=c_{\mathcal{M}}$ is satisfied.

In Fig.~\ref{fig:optimize-ETm} we present simulation results in scenarios for the \textit{TVCM} (Fig.~\ref{fig:optimize-ETm-TVCM}) and Cabspotting (Fig.~\ref{fig:optimize-ETm-Cabspotting}) traces. Content popularity ($P_{p}(n)$) follows a \textit{Zipf} distribution with $n\leq 30$ and exponent $\alpha=\{1,2,3\}$. The availability constraint is set to $c_{\mathcal{M}}=10$. It can be seen that the optimal delay $E[T_{\mathcal{M}}]$ is achieved for $k = 0.5$, as Result~\ref{result:optimal-delay} predicts.
\begin{figure}
\centering
\subfigure[TVCM]{\includegraphics[width=0.49\linewidth]{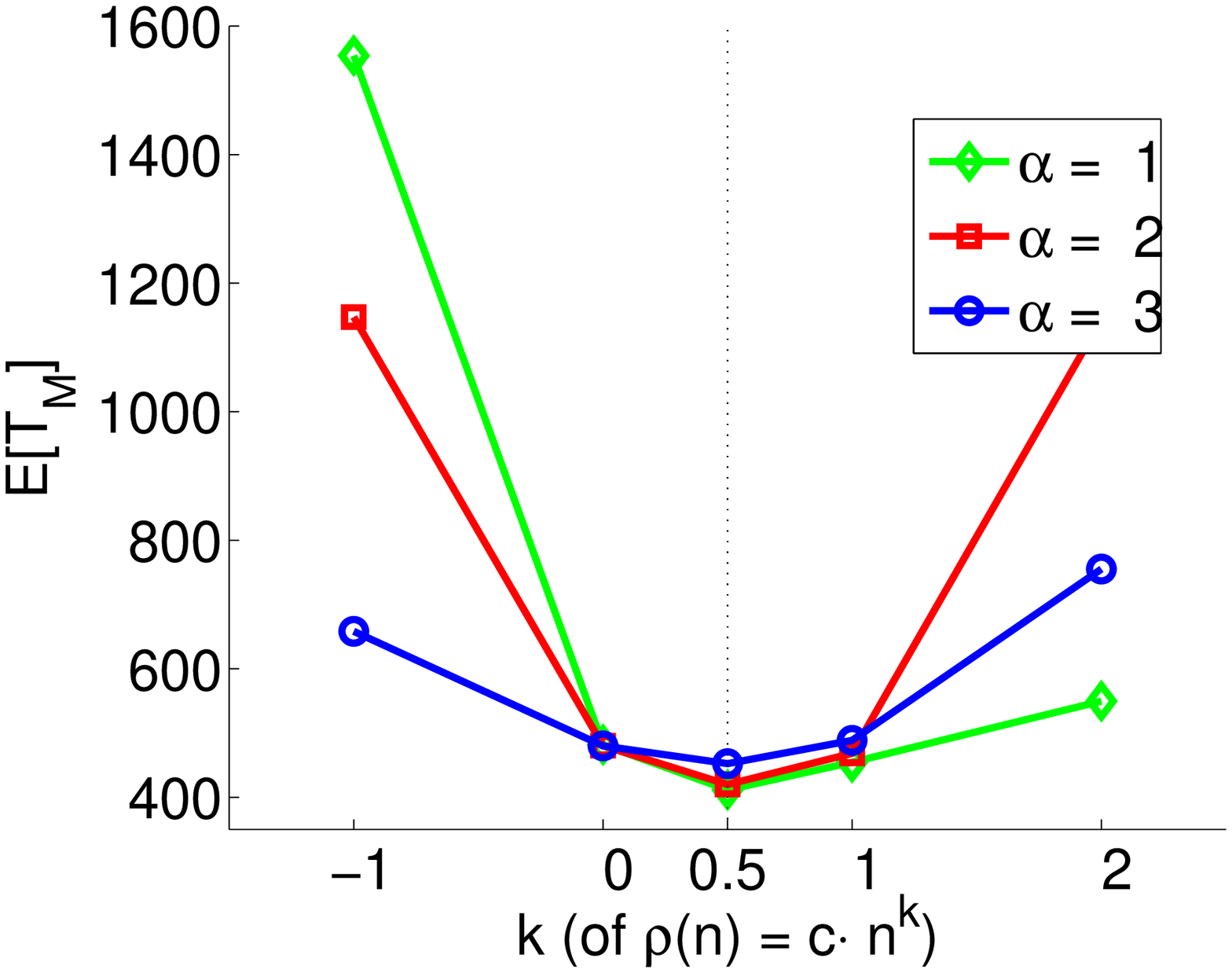}\label{fig:optimize-ETm-TVCM}}
\subfigure[Cabspotting]{\includegraphics[width=0.49\linewidth]{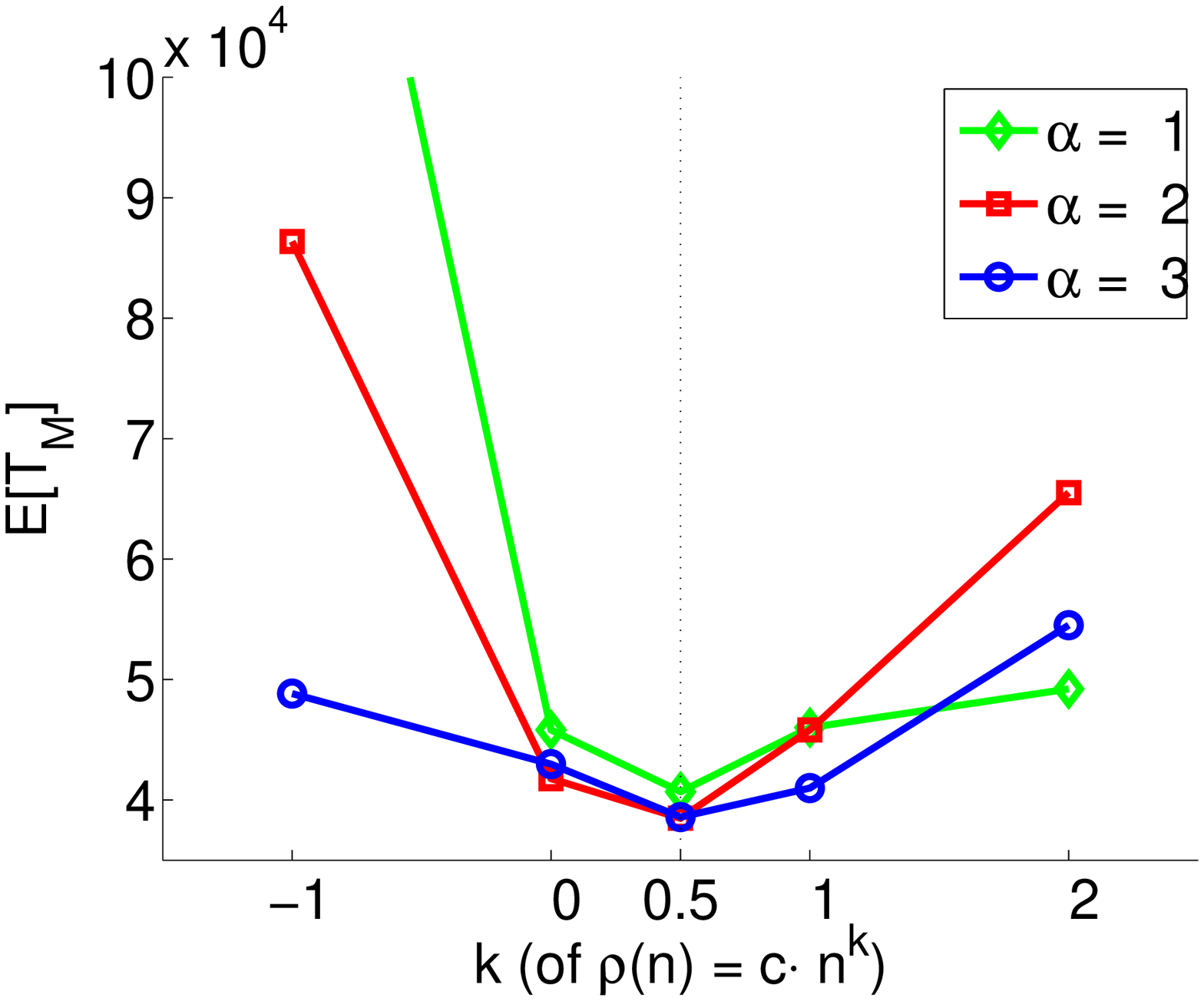}\label{fig:optimize-ETm-Cabspotting}}
\caption{Content access delay $E[T_{\mathcal{M}}]$ of different allocation policies $\rho(n) = c_{k}\cdot n^{k}$, where $c_{k} = \frac{c_{\mathcal{M}}}{E_{p}[n^{k}]}$.}
\label{fig:optimize-ETm}
\end{figure}

\subsubsection{Case 2: $QoS$ constraints}
To evaluate the performance of the allocation function $\rho(n)$ that follows after solving \eq{eq:optimization-problem-Roff-final-expression} (i.e. \textit{optimal} allocation), we compare the \textit{offloading ratio} $R_{off}$ it achieves with the offloading ratios of the following policies:\\
\textit{\textbf{Random}}: We randomly select a content and give a copy of it to a node. We repeat $M\cdot c_{\mathcal{M}}$ times.\\
\textit{\textbf{Square Root}}: We select $\rho(n) \propto \sqrt{n}$ (i.e. the allocation that achieves the minimum expected delay $E[T_{M}]$).\\
\textit{\textbf{Log}}: We select $\rho(n) \propto \log{n}$.

\textit{Random} policy has been used in related work as a baseline~\cite{multiple-offloading} and \textit{square root} policy is the optimal policy when the metric of interest is the content access delay (Section~\ref{sec:case-study-ETm}). Finally, we observed that the \textit{optimal} policy (\eq{eq:optimization-problem-Roff-final-expression}), in the scenarios considered, allocated copies only to the $10\%-20\%$ highest popularity contents. The \textit{log} policy allocates in a similar manner the copies (e.g. no copies to contents with low popularity).

Simulation results on the \textit{SLAW} and \textit{Infocom} scenarios are presented in Fig.~\ref{fig:optimize-Roff-SLAW-TTL0_2} and~\ref{fig:optimize-Roff-Infocom-TTL0_5}, respectively. The parameters in these scenarios are: $M=50$ messages, $P_{p}\sim Zipf$ with $n\in[1,30]$ and $\alpha=1$, total copies $M\cdot c_{\mathcal{M}}=\{50 ,100\}$. As it can be seen our \textit{optimal} policy (leftmost bar) achieves the highest offloading ratio $R_{off.}$. The random policy is clearly inferior than the others. Between \textit{square root} and \textit{log} policies, it is the latter that achieves better performance. These results indicate that, to maximize $R_{off.}$, it is better to allocate the available resources only for popular contents, and serve the non-popular exclusively through the infrastructure.

\begin{figure}
\centering
\subfigure[SLAW, ~~$TTL=530$]{\includegraphics[width=0.49\linewidth]{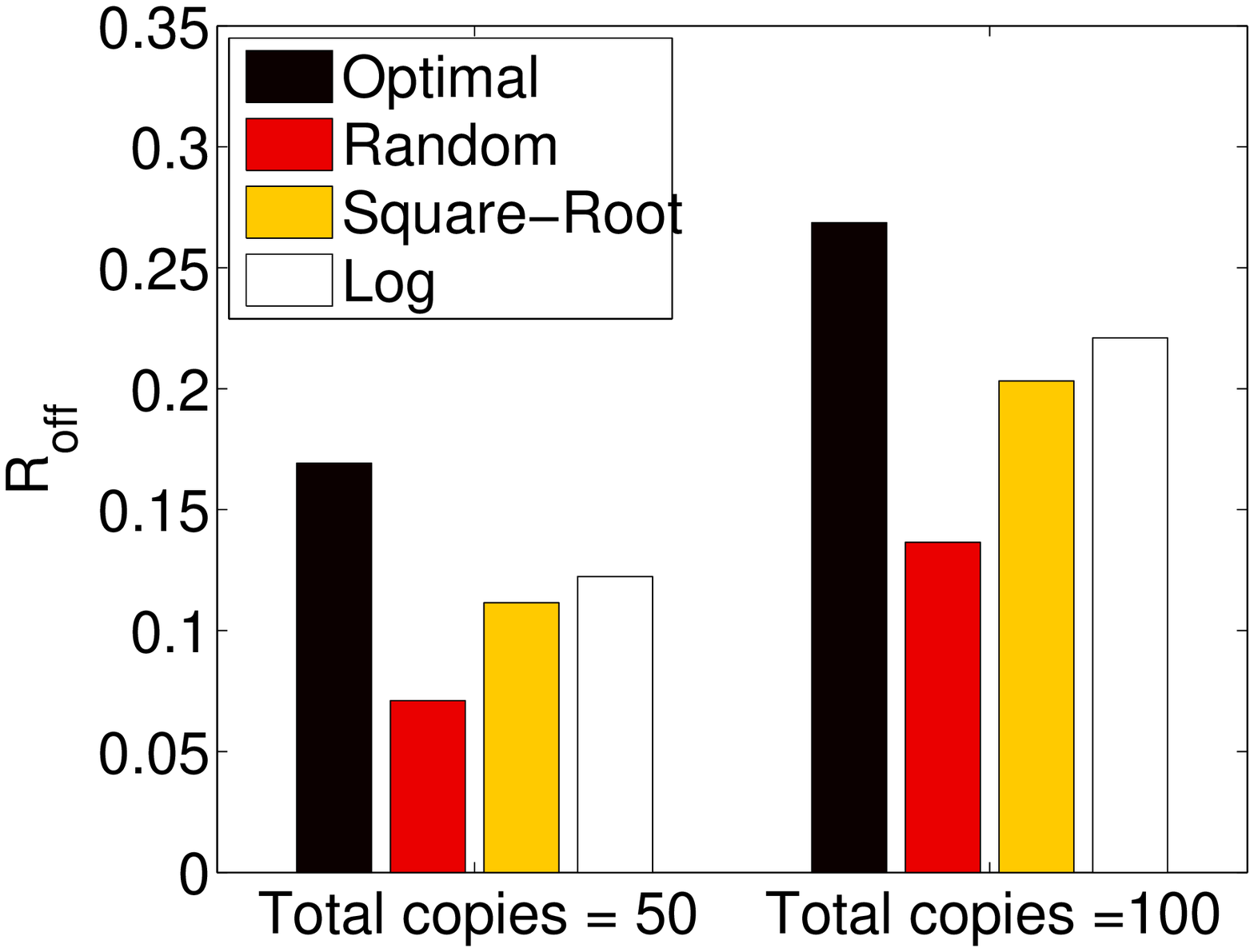}\label{fig:optimize-Roff-SLAW-TTL0_2}}
\subfigure[Infocom, ~~$TTL=10000$]{\includegraphics[width=0.49\linewidth]{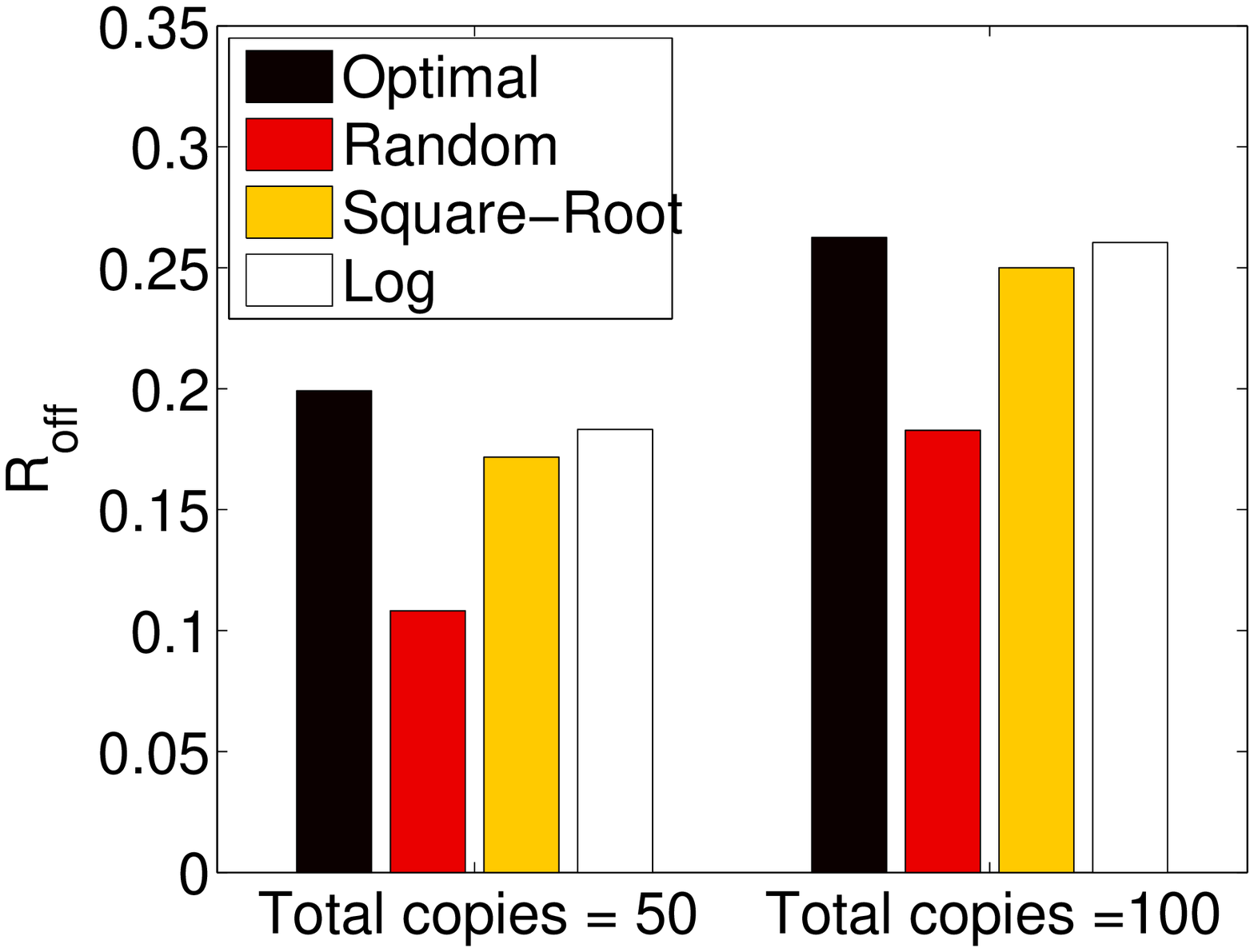}\label{fig:optimize-Roff-Infocom-TTL0_5}}
\caption{Offloading Ratio $R_{off.}$ of different allocation policies $\rho(n)$.}
\label{fig:optimize-Roff}
\end{figure}

\subsection{\revisionRed{Extensions and Discussion}}

\revisionRed{
As a performance evaluation extension, we discuss here some implementation issues for challenging mobile data offloading scenarios, where the knowledge of content popularity and node mobility is limited. We investigate two practical system designs and their performance, and how our theory can be applied in these -much different- scenarios. 
}

\revisionRed{
We believe this section is an initial step towards extending our base framework for more generic settings, and provides further insights for a system implementation.
}

\revisionRed{
\textbf{Popularity-blind system.} We first consider the scenario where the cellular network is not aware of the popularity of the contents. In this case, the options of a system are either to (a) treat every content as equal, following a \textit{uniform} (or, equivalently, a \textit{random}) policy and assigning equal number of holders for each content, or (b) try to estimate online the popularity in order to make a more careful holder assignment.
}

\revisionRed{
To this end, we propose a simple holder assignment algorithm for popularity-blind systems (\textit{no QoS} case), and compare it against the \textit{uniform} policy. Our solution, Algorithm~\ref{alg:heuristic}, combines Result~\ref{result:optimal-delay} (i.e., the optimal holder assignment when content popularity is known) and a simple online popularity estimation heuristic.
}

\revisionRed{
\begin{algorithm}[h]
\caption{Popularity-blind Mobile Data Offloading}
\begin{algorithmic}[1]
\State $\widehat{\Np}=1, \forall\mathcal{M}\in\textbf{M}$ \Comment Set all popularities equal to $1$.
\State $\mathcal{H}\leftarrow $ select\_holders $\left(Result~\ref{result:optimal-delay}~,~c_{\mathcal{M}}~,~\widehat{\Np}\right)$
\For {each content delivery of $\mathcal{M}$}
	\State $\widehat{\Np} = \widehat{\Np} +1$
	\State $\mathcal{H}\leftarrow $ update\_holders $\left(Result~\ref{result:optimal-delay}~,~c_{\mathcal{M}}~,~\widehat{\Np}\right)$
\EndFor
\end{algorithmic}
\label{alg:heuristic}
\end{algorithm}
}

\revisionRed{
Specifically, in Algorithm~\ref{alg:heuristic}, we initially set (line~1) the popularity of all contents equal to $1$, i.e. $\widehat{\Np}=1,~\forall\mathcal{M}\in\textbf{M}$. Then, an equal number of holders is assigned per content (line~2); Result~\ref{result:optimal-delay} gives $c_{\mathcal{M}}$ holders per content when all contents are equally popular. Every time a requester meets a holder and gets the content, the system is informed by the holder or the requester (e.g., with an ACK message) and the estimated popularity of the content is incremented by $1$ (line~4). Finally, the cellular provider updates regularly (once per a time-window, a certain number of content deliveries, etc.) the holder assignment (i.e., assigning/releasing holders) based on Result~\ref{result:optimal-delay} and using the latest estimated popularity values (line~5).
}

\revisionRed{
In Fig.~\ref{fig:unknown-popularity} we compare our \textit{heuristic} approach (Algorithm~\ref{alg:heuristic}) with the \textit{uniform} holder assignment policy, in synthetic mobility scenarios $f_{\lambda}\sim Gamma (\mu_{\lambda}=1, CV_{\lambda}=1)$ with content popularity $P_{p}(n) = Zipf(n\in[1,30]~,~\alpha = 1)$. As shown, the proposed algorithm leads to lower delivery delays than the \textit{uniform} policy. Moreover, it can be seen that, even without having any knowledge in advance about popularity patterns (worst-case scenario) and using a simple mechanism, we can achieve a performance close to the optimal (of a popularity-aware mechanism, i.e., Result~\ref{result:optimal-delay}). 
}

\revisionRed{
We observed similar behavior in a number of different simulation scenarios. The performance of Algorithm~\ref{alg:heuristic} is always better than the \textit{uniform} policy; the distance from the optimal case depends on the scenario, but is consistently close to it.
}

\revisionRed{
\textbf{Temporal mobility patterns.} As discussed earlier, considering only some average mobility characteristics (Section~\ref{sec:mobility-model}) not only facilitates analysis, but also, the implementation of real systems. In some scenarios though, a more detailed approach might be necessary. As an example, we consider here cases where a content distribution experiences long delays, so that temporal mobility characteristics come into play, i.e. the pairwise contact rates $\lambda_{ij}$ might change before delivery is completed. 
}

\revisionRed{
In particular, we assume a scenario composed of two alternating time windows of constant contact rates: in each time window $tw_{1}$ and $tw_{2}$ the contact rate of a node pair $\{i,j\}$ is $\lambda_{ij}^{(1)}$ and $\lambda_{ij}^{(2)}$, respectively. This could be the case, e.g., of different day/night node mobility patterns. We draw the corresponding pairwise contact rates from two \textit{Gamma} distributions $f_{\lambda}^{(1)}(\lambda)$ and $f_{\lambda}^{(2)}(\lambda)$, with $\mu_{\lambda}^{(1)}=1$ and $\mu_{\lambda}^{(2)}=5$. Content popularity patterns are $P_{p}(n) = Zipf(n\in[1,100]~,~\alpha = 2)$.
}

\revisionRed{
To investigate the effects of these temporal characteristics, we compare three mobile data offloading (with \textit{QoS}, $TTL=0.17$) mechanisms:\\
{\textit{Optimal (average)}}: The system is aware of the mobility patterns in \textit{both} time windows. The holder assignment is done based on the solution of \eq{eq:optimization-problem-Roff-final-expression}, with $\mu_{\lambda} = \frac{\mu_{\lambda}^{(1)}+\mu_{\lambda}^{(2)}}{2}$.\\
{\textit{Optimal (window-based)}}: The system is aware of the mobility patterns \textit{only} of the window in which the content distribution begins. The holder assignment is done based on the solution of \eq{eq:optimization-problem-Roff-final-expression}, with $\mu_{\lambda} =\mu_{\lambda}^{(1)}$ or $\mu_{\lambda} =\mu_{\lambda}^{(2)}$.\\
{\textit{Log}}: This mechanism is presented in Section~\ref{sec:offloading-evaluation}.
}

\revisionRed{
We present the simulation results in Fig.~\ref{fig:varying-lambda}. We can see that the \textit{Optimal (window-based)} mechanism, where the knowledge of mobility patterns is limited to only one time-window, does not achieve an $R_{off.}$ as high as in the case of the \textit{Optimal (average)} mechanism that has a complete view of mobility. Nevertheless, having even a limited information about the mobility patterns, is beneficial: as shown in Fig.~\ref{fig:varying-lambda}, \textit{Optimal (window-based)} performs always better than the \textit{Log} policy, which has been shown to achieve the best performance among the mobility oblivious policies (see Fig.~\ref{fig:optimize-Roff}).
}

\revisionRed{
Finally, as the window size increases (from the left to the right set of bars), the difference in the performance between the two \textit{Optimal} approaches diminishes. The reason is that a larger part of the content distribution process takes place within a single time window, and thus the extra knowledge of the \textit{Optimal (average)} mechanism adds less value to the prediction accuracy. This observation further supports our argument (see Section~\ref{sec:mobility-model}) that considering only a few average statistics is a good choice when there is a time-scale separation between content delivery and temporal mobility characteristics.
}

\begin{figure}
\subfigure[Popularity-blind offloading]{\includegraphics[width=0.49\linewidth]{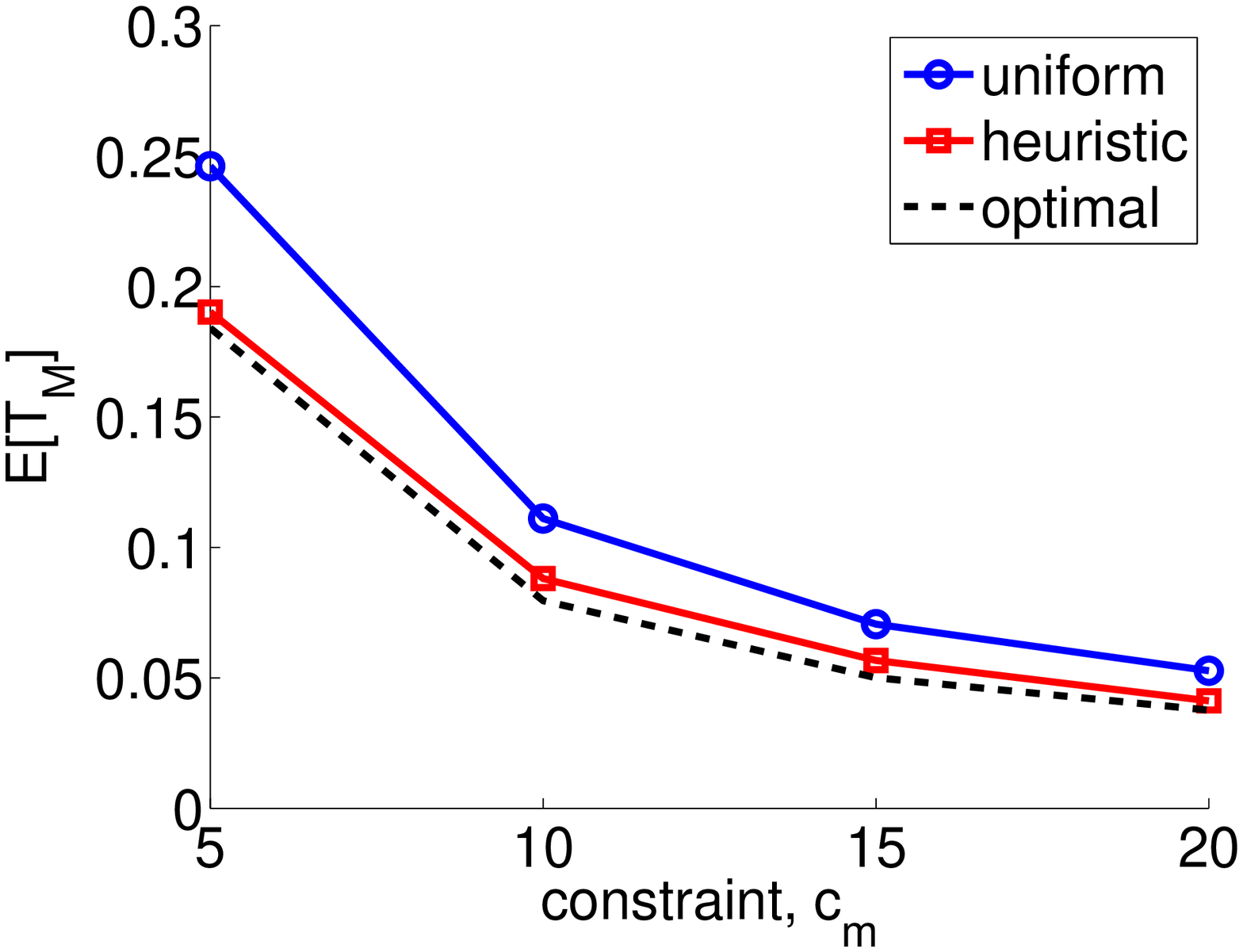}\label{fig:unknown-popularity}}
\subfigure[Temporal mobility patterns]{\includegraphics[width=0.49\linewidth]{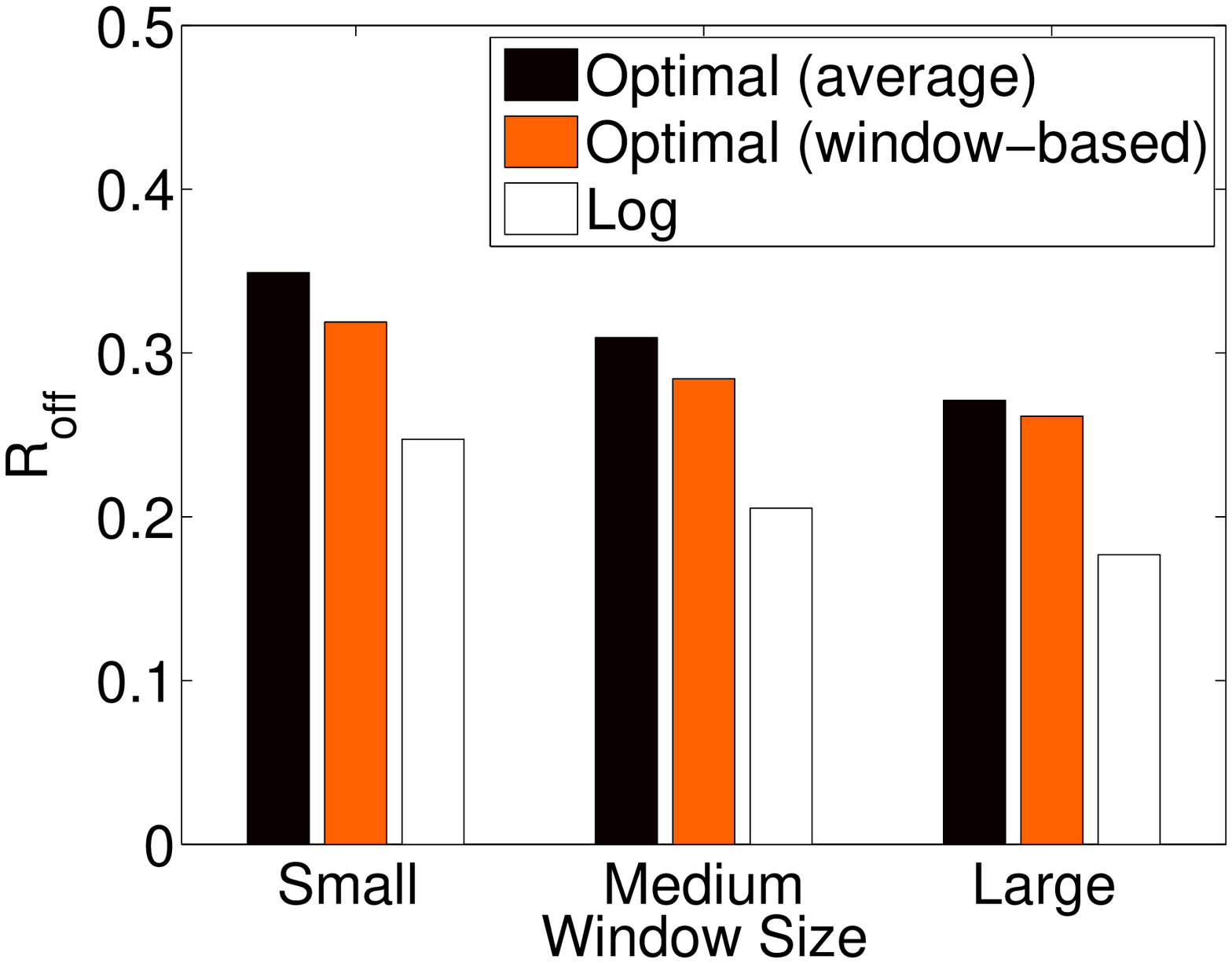}\label{fig:varying-lambda}}
\caption{\revisionRed{Simulation results for scenarios where (a) content popularity is not known in advance, and (b) mobility patterns drastically change within the content delivery time-window. }}
\label{fig:offloading-extensions}
\end{figure}

\section{Related Work}\label{sec:related-work}
Content-centric applications were introduced in opportunistic networking under the \textit{publish - subscribe} paradigm~\cite{Yoneki-publish-subscribe-dtn,podcasting,Costa-publish-subscribe-dtn,contentplace}, for which several data dissemination techniques have been proposed. In~\cite{Yoneki-publish-subscribe-dtn}, authors propose a mechanism that identifies social communities and the nodes-``hubs'', and builds an overlay network between them in order to efficiently disseminate data. SocialCast~\cite{Costa-publish-subscribe-dtn} based on information about nodes interests, social relationships and movement predictions, selects the set of holders. Similarly to the above approaches, ContentPlace~\cite{contentplace} uses both community detection and nodes social relationships information, to improve the performance of the content distribution. 

Under a different setting,~\cite{Gao-user-centric-DTN,CEDO} study content sharing mechanisms with limited resources (e.g. buffer sizes, number of holders). In~\cite{Gao-user-centric-DTN}, authors analytically investigate the data dissemination cost-effectiveness tradeoffs, and propose techniques based on contact patterns (i.e. $\lambda_{ij}$) and nodes interests. Similarly, CEDO~\cite{CEDO} aims at maximizing the total content delivery rate: by maintaining a utility per content, nodes make appropriate drop and scheduling decisions. 

\red{Some further modeling and analytic techniques for content-centric opportunistic networking include~\cite{gossip-age,float-content-infocom11}. In~\cite{gossip-age}, authors use a community mobility model and an analysis based on mean-field techniques to study an application of content updates, and derive results for the distribution of content age under different settings. \cite{float-content-infocom11} considers an application for local dissemination of contents and derives criticality conditions under which the content distribution (floating) is viable.}

Recently, further novel content-centric application have been proposed, like location-based applications~\cite{MobComp-next-decade,Ott-oppnet-applications} and mobile data offloading~\cite{offloading-wowmom11,Hui-Offloading, multiple-offloading}. The latter category, due to the rapid increase of mobile data demand, has attracted a lot of attention. In the setting of~\cite{offloading-wowmom11}, content copies are initially distributed (through the infrastructure) to a subset of mobile nodes, which then start propagating the contents epidemically. Differently, in~\cite{Hui-Offloading} the authors consider a limited number of holders, and study how to select the best holders-target-set for each message. In~\cite{multiple-offloading}, the same problem is considered, and (centralized) optimization algorithms are proposed that take into account more information about the network: namely, size and lifetimes of different contents, and interests, privacy policies and buffer sizes of each node.

In the majority of previous studies, although node interests and content popularity are taken into account, the focus has been on the algorithms and the applications themselves. We believe that our study complements existing work, by providing a common analytical framework for a number of these approaches that can be used both for predicting the performance of proposed schemes, as well as proposing improved ones.

\section{Conclusion}\label{sec:conclusion}
The increasing number of mobile devices and traffic demand, renders content-centric applications through opportunistic communication very promising. Hence, motivated by the lack of a common analytical framework, we modeled and analyzed the effects of content popularity / availability patterns in the performance of content-centric mechanisms.

As a part of future work we intend to study, in more detail, extensions of our model and to investigate further characteristics of content traffic patterns, like \textit{traffic locality} in location based social networks, and their performance effects.

\ifCLASSOPTIONcaptionsoff
  \newpage
\fi



\begin{IEEEbiography}[{\includegraphics[width=1in,height=1.25in,clip,keepaspectratio]{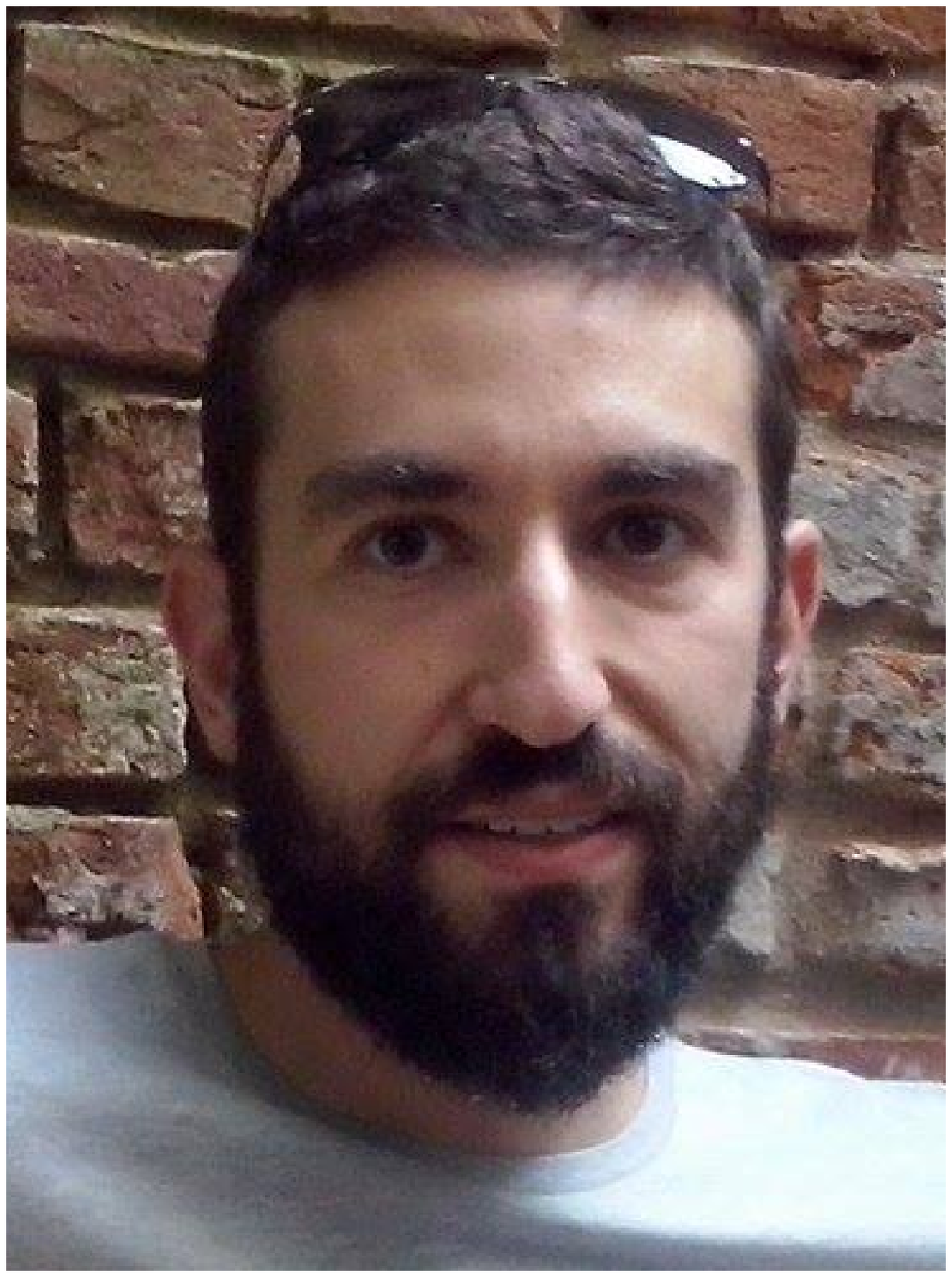}}]{Pavlos Sermpezis}
received the Diploma in Electrical and Computer Engineering from the Aristotle University of Thessaloniki, Greece, and a PhD in Computer Science and Networks from EURECOM, Sophia Antipolis, France. He is currently a post-doctoral researcher at FORTH, Greece. His main research interests are in modeling and performance analysis for mobile-to-mobile communications, and interdomain routing for the Internet.
\end{IEEEbiography}
\begin{IEEEbiography}[{\includegraphics[width=1in,height=1.25in,clip]{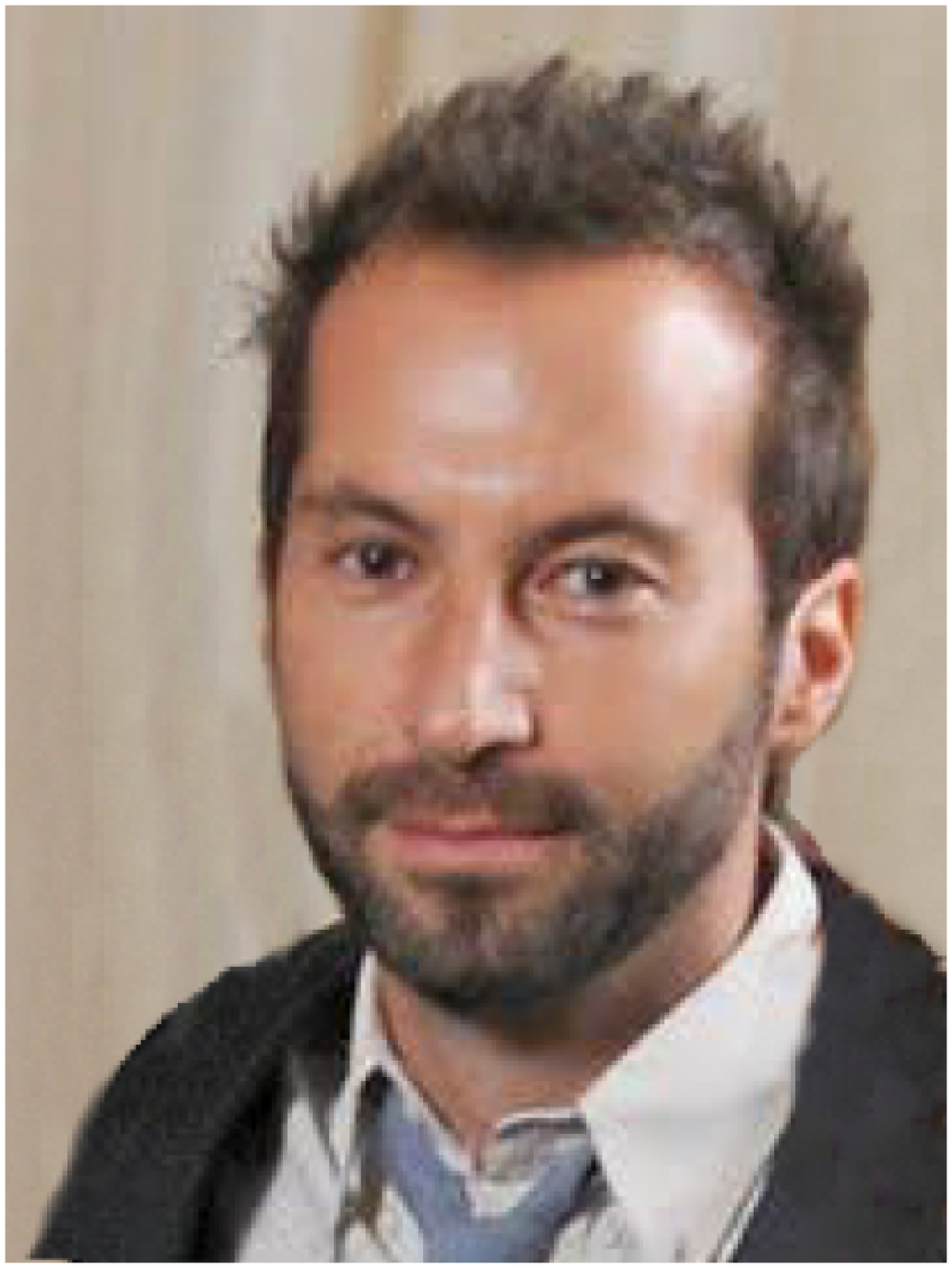}}]{Thrasyvoulos Spyropoulos}
received the Diploma in Electrical and Computer Engineering from the National Technical University of Athens, Greece, and a Ph.D degree in Electrical Engineering from the University of Southern California. He was a post-doctoral researcher at INRIA and then, a senior researcher with the Swiss Federal Institute of Technology (ETH) Zurich. He is currently an Assistant Professor at EURECOM, Sophia-Antipolis. He is the recipient of the best paper award in IEEE SECON 2008, and IEEE WoWMoM 2012.
\end{IEEEbiography} 
\vfill

\appendices
\section{Proof of Result~\ref{R:ATD}}\label{appendix:proof-R:ATD}

\begin{proof}
To calculate the average performance, we need to modify the previous analysis as following: Consider a content $\mathcal{M}$ of initial popularity $\Np(0)=n$ and availability $\Na(0)=m$, i.e. initially $n$ nodes are looking for the content and $m$ nodes hold the content. When the first requester access the content, the number of holders will increase to $m+1$ and the remaining requesters will be $n-1$. Building a Markov Chain as in Fig~\ref{fig:markov-chain}, where each state denotes the number of holders, it can be shown for the expected delay of moving from state $m+k$ to state $m+k+1$, $k\in[0,n-1]$, that it holds
\begin{equation}\label{eq:ET-step}
E[T_{k,k+1}] \approx \frac{1}{(m+k)\cdot (n-k)\cdot \mu_{\lambda}} 
\end{equation}
where $m+k$ are the nodes holding the content, $n-k$ the remaining requesters and $\mu_{\lambda}$ the mean contact rate.

From the above analysis, it follows straightforward that the expected time till the first requester to access the message is $E[T^{1}]=E[T_{0,1}]$ and till the $\ell^{th}$ requester to access it is 
\begin{equation}\label{eq:ET-ell}
E[T^{\ell}]=\sum_{k=0}^{\ell-1}E[T_{k,k+1}]
\end{equation}

Let us now define the sum of delays $E[T^{\ell}]$ (i.e. delivery delays for each requester) for a message $\mathcal{M}$ with initial availability $\Na(0)=m$ and initial popularity $\Np(0)=n$, as:
\begin{equation} 
S(T_{\mathcal{M}}|m,n)= \sum_{\ell=1}^{n} E[T^{\ell}|m,n]
\end{equation}
From \eq{eq:ET-step} and \eq{eq:ET-ell}, we can write for $S(T_{\mathcal{M}}|m,n)$:
\begin{align}
S(T_{\mathcal{M}}|m,n)
&\approx  \sum_{\ell=1}^{n}\sum_{k=0}^{\ell-1}\frac{1}{(m+k)\cdot (n-k)\cdot \mu_{\lambda}} \nonumber\\
& =\sum_{k=0}^{n-1}(n-k)\cdot \frac{1}{(m+k)\cdot (n-k)\cdot \mu_{\lambda}} \nonumber\\
& =\frac{1}{ \mu_{\lambda}}\cdot\sum_{k=0}^{n-1}\frac{1}{m+k} \nonumber\\
& =\frac{1}{ \mu_{\lambda}}\cdot\sum_{k=m}^{m+n-1}\frac{1}{k}
\end{align}
and using the approximation of the harmonic sum\footnote{$\sum_{k=1}^{N}\approx \ln(N)+\gamma +O\left(\frac{1}{N}\right)$, where $\gamma$ is the \textit{Euler-Mascheroni} constant.}, we get
\begin{align}\label{eq:Sum-Tm-m,n}
S(T_{\mathcal{M}}|m,n)\approx \frac{1}{ \mu_{\lambda}}\cdot \ln\left(1+\frac{n}{m-1}\right)\approx \frac{1}{ \mu_{\lambda}}\cdot \ln\left(1+\frac{n}{m}\right)
\end{align}

Averaging over all the content in the network, we can write for the expected content access delay:
\begin{align}
E[T_{\mathcal{M}}] = \frac{\sum_{\mathcal{M}}S(T_{\mathcal{M}})}{\sum_{\mathcal{M}}\Np}
\end{align}
or, since (i) (by definition) there are $M\cdot P_{p}(n)$ contents in the network, and (ii) we do not differentiate between contents with the same popularity/availability:
\begin{align}
E[T_{\mathcal{M}}] &= \frac{\sum_{n}S(T_{\mathcal{M}}|n)\cdot (M\cdot P_{p}(n))}{\sum_{\mathcal{M}}n\cdot(M\cdot P_{p}(n))}= \frac{\sum_{n}S(T_{\mathcal{M}}|n) \cdot P_{p}(n)}{E_{p}[n]}\nonumber\\
&= \frac{\sum_{n}S(T_{\mathcal{M}}|n,m)\cdot g(m|n) \cdot P_{p}(n)}{E_{p}[n]}\nonumber\\
&\approx \frac{\sum_{n}\frac{1}{ \mu_{\lambda}}\cdot \ln\left(1+\frac{n}{m}\right)\cdot g(m|n) \cdot P_{p}(n)}{E_{p}[n]}
\end{align}
where in the last line we substituted from \eq{eq:Sum-Tm-m,n}.

We can further use Jensen's inequality (since the function $h(x) = \ln\left(1+\frac{n}{x}\right)$ is convex) or the respective approximation, and finally write:
\begin{align}
E[T_{\mathcal{M}}] \approx \frac{1}{\mu_{\lambda}\cdot E_{p}[n]}\cdot E_{p}\left[ \ln\left(1+\frac{n}{\overline{g}(n)}\right)\right]
\end{align}
which proves the result.
%
\end{proof}

\section{Proof of Result~\ref{R:MDA} and Example}\label{appendix:R:MDA}

\begin{proof}
Def.~\ref{def: Mobility Dependent Allocation} says that \textit{who} holds a content and \textit{who} is interested in it is not independent of their mobility patterns. The contact rates between the requester of a content and the holders of it, are not distributed with the contact rates distribution $f_{\lambda}(\lambda)$, since the requesters-holders contact rates are mobility dependent. It can be shown that the requesters-holders contact rates are distributed as~\cite{pavlos-TMC-traffic}
\begin{equation}
 f_{\pi}(\lambda) = \frac{1}{E_{\lambda}[\pi(\lambda)]}\cdot \pi(\lambda)\cdot f_{\lambda}(\lambda)
\end{equation}
Hence, \eq{eq:definition-fml} and \eq{eq:mean-ml-x} need to be modified as:
\begin{equation}\label{eq:definition-fmpi}
\XM\sim f_{m\pi}(x) = \left(f_{\pi}\ast f_{\pi} \cdots \ast f_{\pi} \right)_{m}
\end{equation}
and
\begin{equation}\label{eq:mean-mpi-x}
E[\XM|\Na=m] = E_{m\pi}[x] = m\cdot \frac{E_{\lambda}[\lambda\cdot \pi(\lambda)]}{E_{\lambda}[\pi(\lambda)]}= m\cdot \mu_{\lambda}^{(\pi)}
\end{equation}
Then, it can be easily seen that following the same analysis, we get the same expressions as in Theorems~\ref{thm:lower-bound-ETm} and~\ref{thm:upper-bound-P-Tm-TTL} and Result~\ref{R:ATD} where, now, the mean contact rate $\mu_{\lambda}$ is replaced by the \textit{mean mobility dependent requesters-holders contact rate} $\mu_{\lambda}^{(\pi)}$.
\end{proof}

\paragraph*{Example Scenario} For each content $\mathcal{M}$, its holders are selected taking into account their contact rates with the requesters with the following mechanism: Each node $i$ candidate to be a holder is assigned a weight $w_{i} = \prod_{j\in\Cp}\lambda_{ij}$. Then, each of them is selected to be one of the $\Na$ holders with probability $p_{i} = \frac{w_{i}}{\sum_{i}w_{i}}$. Now, for the node pair $\{i,j\}$ ($i\in\Ca$, $j\in\Cp$) it holds that
\begin{equation}
 \pi_{ij} = \frac{w_{i}}{\sum_{i}w_{i}}= \frac{\prod_{k\in\Cp}\lambda_{ik}}{\sum_{i}\prod_{k\in\Cp}\lambda_{ik}}=\frac{\lambda_{ij}\cdot \prod_{k\in\Cp\setminus\{j\}}\lambda_{ik}}{\sum_{i}\prod_{k\in\Cp}\lambda_{ik}}
\end{equation}
for which, when the node popularity $\Np=|\Cp|$ is large enough, we can write
\begin{equation}
 \pi_{ij} \approx \frac{\lambda_{ij}\cdot c_{1}}{c_{2}}
\end{equation}
where $c_{1},c_{2}$ take approximately the same value $\forall i,j$, i.e. $\pi(\lambda) = c\cdot \lambda$, $c=\frac{c_{1}}{c_{2}}$. Substituting $\pi(\lambda)$ in Result~\ref{R:MDA}, gives
\begin{equation}
 \mu_{\lambda}^{(\pi)} = \frac{E_{\lambda}[\lambda\cdot \pi(\lambda)]}{E_{\lambda}[\pi(\lambda)]}=\frac{E_{\lambda}[\lambda^{2}]}{E_{\lambda}[\lambda]} = \mu_{\lambda}\cdot(1+CV_{\lambda}^{2})
\end{equation}

\section{Minimum of Pareto distributed random variables}\label{appendix:min-pareto}

\begin{proof}
For the random variable $T_{\mathcal{M}} = \min_{i\in\mathcal{C}_{a}}\lbrace T_{ij}^{(r)}\rbrace$, where each $T_{ij}^{(r)}$ is a random variable distributed with a Pareto distribution with scale parameter $t_{0}$ and shape parameter $\alpha_{ij}$, it holds that:
\begin{align}
P\lbrace T_{\mathcal{M}}>t\rbrace 	&= \prod_{i\in\mathcal{C}_{a}} P\lbrace T_{ij}^{(r)}>t\rbrace = \prod_{i\in\mathcal{C}_{a}} \left(\frac{t_{0}}{t}\right)^{\alpha_{ij}}\nonumber\\
									&= \left(\frac{t_{0}}{t}\right)^{\sum_{i\in\mathcal{C}_{a}} \alpha_{ij}}
\end{align}
which means that $T_{\mathcal{M}}$ follows a Pareto distribution with scale and shape parameters $t_{0}$ and $A_{\mathcal{M}}=\sum_{i\in\mathcal{C}_{a}} \alpha_{ij}$, respectively. 
\end{proof}

\section{Proofs for the performance metrics expressions of the Pareto case}\label{appendix:pareto-expressions-proofs}
\subsection{Content Access Delay}
\begin{proof}
The expectation of an (American) Pareto distributed random variable ($Pareto(\alpha, t_{0})$) is $\frac{t_{0}}{\alpha-1}$. Hence, in the derivation of \eq{eq:expected-Ta-generic}, one only needs to change the integral in the last equality as:
\begin{equation}\label{eq:expected-Ta-generic-pareto}
E[T_{\mathcal{M}}] = \sum_{m} \int \frac{t_{0}}{x-1}\cdot f_{m\alpha}(x)dx \cdot P_{a}^{req.}(m)
\end{equation}
Substituting $P_{a}^{req.}(m)$ from Lemma~\ref{thm:Pint-availability-generic} and proceeding as in the exponential case, we subsequently get:
\begin{align}
&E[T_{\mathcal{M}}] = \sum_{m} \int \frac{t_{0}}{x-1}\cdot f_{m\alpha}(x)dx \cdot \frac{E_{p}[n\cdot g(m|n)]}{E_{p}[n]}\nonumber\\
&=\frac{t_{0}}{E_{p}[n]}\cdot E_{p}\left[n\cdot \sum_{m} E_{m\alpha}\left[\frac{1}{x-1}\right]\cdot g(m|n)\right]
\end{align}
which is the exact expression for $E[T_{\mathcal{M}}]$ in Table~\ref{table:pareto-expressions}.

Applying \textit{Jensen's inequality} for the convex function $h(x)=\frac{1}{x-1}$, gives:
\begin{equation}
E_{m\alpha}\left[\frac{1}{x-1}\right]\geq \frac{1}{m\cdot\mu_{\alpha}-1}
\end{equation}
and, thus:
\begin{align}
E[T_{\mathcal{M}}] &\geq \frac{t_{0}}{E_{p}[n]}\cdot E_{p}\left[n\cdot \sum_{m} \frac{1}{m\cdot\mu_{\alpha}-1}\cdot g(m|n)\right]\nonumber\\
&=\frac{t_{0}}{E_{p}[n]}\cdot E_{p}\left[n\cdot E_{g}\left[\frac{1}{m\cdot\mu_{\alpha}-1}\right]\right]\nonumber\\
&\geq \frac{t_{0}}{E_{p}[n]}\cdot E_{p}\left[n\cdot \frac{1}{\overline{g}(n)\cdot\mu_{\alpha}-1}\right]
\end{align}
where for the last line we applied \textit{Jensen's inequality} for the expectation $E_{g}\left[\frac{1}{m\cdot\mu_{\alpha}-1}\right]$.
\end{proof}

\subsection{Content Access Probability}
\begin{proof}
In the Pareto case, the integral in \eq{eq:probability-Ta-generic} changes as: $\int \left(\frac{t_{0}}{t_{0}+TTL}\right)^{x}\cdot f_{m\alpha}(x)dx$, for $TTL\geq t_{0}$, because for a Pareto random variable $x\sim Pareto(\alpha, t_{0})$ it holds that $P\{x\leq TTL\}=1-\left(\frac{t_{0}}{t_{0}+TTL}\right)^{\alpha}$. Following the same methodology as before and observing that the function $h(x) = \left(\frac{t_{0}}{t_{0}+TTL}\right)^{\alpha}$ is convex, the expressions of Table~\ref{table:pareto-expressions} follow similarly.
\end{proof}

\section{Proof of Result~\ref{result:optimal-delay}}\label{appendix:proof-result:optimal-delay}
\begin{proof}
Using as an approximation for $E[T_{\mathcal{M}}]$ the expression of Theorem~\ref{thm:lower-bound-ETm}, we can write
\begin{equation}
 \textstyle E[T_{\mathcal{M}}]= \frac{1}{\mu_{\lambda}\cdot E_{p}[n]}\cdot E_{p}\left[\frac{n}{\overline{g}(n)}\right]\nonumber
\end{equation}
Jensen's inequality used in \eq{eq:jensen-g}, becomes equality when $g(m|n)$ is deterministic. This suggests that among all the functions $g(m|n)$ with the same average value $\overline{g}(n)$, the minimum delay can be achieved in the case: $\rho(n) = \overline{g}(n)$. Thus, the $E[T_{\mathcal{M}}]$ minimization problem becomes equivalent to 

\begin{footnotesize}
\begin{align}\label{eq:equivalent-min-ETm}
\min\{E_{p}\left[\frac{n}{\rho(n)}\right]\} = \sum_{n} \frac{n}{\rho(n)}\cdot P_{p}(n)=\sum_{n} \frac{n}{\rho_{n}}\cdot P_{p}(n)
\end{align}
\end{footnotesize}
where we expressed the expectation as a sum and denoted $\rho_{n} = \rho(n)$.

Moreover, we can express the content copies constraint as
\begin{equation}\label{eq:equivalent-constraint-ETm}
c_{\mathcal{M}} = \textstyle \frac{\sum_{\mathcal{M}} \Na}{M} = E_{p}[\rho(n)]=\sum_{n} \rho_{n}\cdot P_{p}(n)
\end{equation}
Using \eq{eq:equivalent-min-ETm} and \eq{eq:equivalent-constraint-ETm}, the optimization problem becomes
\begin{equation}\label{eq:delay-optimization-rho(n)}
 \min_{\overline{\rho}}\lbrace \sum_{n} \frac{n}{\rho_{n}}\cdot P_{p}(n)\rbrace~~~~s.t.~~~\sum_{n} \rho_{n}\cdot P_{p}(n)= c_{\mathcal{M}}
\end{equation}
where $\overline{\rho}$ denotes the vector with components $\rho_{n}$. 

The optimization problem of \eq{eq:delay-optimization-rho(n)} is \textit{convex} and, thus, it can be solved with the method of Lagrange multipliers~\cite{practical-optimization-book}. Hence, we need to find the values of $\overline{\rho}$ for which it holds that
\begin{equation}
 \nabla \left(\sum_{n} \frac{n}{\rho_{n}}\cdot P_{p}(n)\right) + \nabla \lambda_{0} \left(\sum_{n} \rho_{n}\cdot P_{p}(n)- c_{\mathcal{M}}\right)=0\nonumber
\end{equation}
where $\lambda_{0}$ is the langrangian multiplier. Here, the constraint $\rho_{n}\geq 0$ needs also to be taken into account. However, it is proved to be an inactive constraint (the solution satisfies it) and thus we omit it at this step for simplicity. Similarly, we assume a large enough network, i.e. always holds $\rho_{n}\leq N$.

The differentiation over $\rho_{n}$ gives
\begin{equation}\label{eq:Na-l0}
 \rho_{n} = \frac{1}{\sqrt{\lambda_{0}}}\cdot\sqrt{n}
\end{equation}
Substituting \eq{eq:Na-l0} in the constraint expression $\sum_{n} \rho_{n}\cdot P_{p}(n)= c_{\mathcal{M}}$ (\eq{eq:delay-optimization-rho(n)}), we can easily get
\begin{equation}\label{eq:sqrt-lambda0}
 \sqrt{\lambda_{0}} = \frac{\sum_{n} \sqrt{n}\cdot P_{p}(n)}{c_{\mathcal{M}}} = \frac{E_{p}[\sqrt{n}]}{c_{\mathcal{M}}}
\end{equation}
Then, substituting \eq{eq:sqrt-lambda0} in \eq{eq:Na-l0}, gives
\begin{equation}\label{eq:optimal-rho(n)}
\rho(n) = \rho_{n} =  \frac{c_{\mathcal{M}}}{E_{p}[\sqrt{n}]} \cdot\sqrt{n}
\end{equation}
Finally, the values of \eq{eq:optimal-rho(n)} satisfy the \textit{Karush-Kuhn-Tucker} conditions, which means that the solution of \eq{eq:optimal-rho(n)} is a global minimum~\cite{practical-optimization-book}.
\end{proof}

\end{document}